\documentclass[11pt]{article}

\newcommand{\comment}[1]{}

\usepackage{epsfig, epstopdf}
\usepackage{amsmath, amsthm, amsfonts, amssymb, latexsym, stmaryrd}

\usepackage{float}
\usepackage{tikz}
\usepackage{subfig}

\newtheorem{lemma}{Lemma}[section]
\newtheorem{theorem}[lemma]{Theorem}
\newtheorem{corollary}[lemma]{Corollary}
\newtheorem{proposition}[lemma]{Proposition}
\newtheorem{claim}[lemma]{Claim}
\newtheorem{observation}[lemma]{Observation}
\newtheorem{definition}[lemma]{Definition}

\newtheorem{example}[lemma]{Example}

\def\squareforqed{\hbox{\rlap{$\sqcap$}$\sqcup$}}
\def\qed{\ifmmode\squareforqed\else{\unskip\nobreak\hfil
\penalty50\hskip1em\null\nobreak\hfil\squareforqed
\parfillskip=0pt\finalhyphendemerits=0\endgraf}\fi}
\newcommand{\reals}{\mbox{$\mathbb R$}}
\newcommand{\rats}{\mathbb Q}
\newcommand{\prats}{\mathbb Q_+}
\newcommand{\nats}{\mbox{$\mathbb N$}}

\newcommand{\pr}{{\mathop{\mathrm{pr}}\nolimits}} 
\newcommand{\maxp}{{\mathop{\mathrm{pr^{\star}}}\nolimits}} 
\newcommand{\cost}{{\mathop{\mathrm{cst}}\nolimits}} 

\newlength{\tablength}
\newlength{\spacelength}
\newcommand{\tabstar}{\hspace*{\tablength}}
\newcommand{\spacestar}{\hspace*{\spacelength}}
\def\obeytabs{\catcode`\^^I=\active}
{\obeytabs\global\let^^I=\tabstar}
{\obeyspaces\global\let =\spacestar}
\newenvironment{display}{\begingroup\obeylines\obeyspaces\obeytabs}{\endgroup}
\newenvironment{prog}{\begin{display}\parskip0pt\sf}{\end{display}}

\title{The structure and topology of rooted weighted trees
modeling layered cyber-security systems}

\author{
{\sl Geir Agnarsson}\thanks{Corresponding author.}\\
Department of Mathematical Sciences\\
George Mason University\\
Fairfax, VA 22030\\
{\tt geir@math.gmu.edu}
\and
{\sl Raymond Greenlaw}\\ 
Department of Cyber Sciences\\
United States Naval Academy\\
Annapolis, Maryland 21402\\
{\tt greenlaw@usna.edu}
\and
{\sl Sanpawat Kantabutra}\\
The Theory of Computation Group\\
Computer Engineering Department\\
Chiang Mai University\\
Chiang Mai, 50200, Thailand\\
{\tt sanpawat@alumni.tufts.edu}
}

\date{\today}

\begin{document}

\maketitle

\begin{abstract}
In this paper we consider a layered-security model in which the
containers and their nestings are given in the form of a rooted tree
$T$\@. A {\em cyber-security model\/} is an ordered three-tuple 
$M = (T, C, P)$ where $C$ and $P$ are multisets of {\em penetration costs\/} 
for the containers and {\em target-acquisition values\/} for 
the prizes that are located within the containers, respectively, 
both of the same cardinality as the set of the non-root vertices of $T$.
The problem that we study is to assign the
penetration costs to the edges and the target-acquisition values to
the vertices of the tree $T$ in such a way that minimizes the total
prize that an attacker can acquire given a limited {\em budget}. 
The attacker breaks into containers starting at the root of $T$ and once a
vertex has been broken into, its children can be broken into by paying
the associated penetration costs.  The attacker must deduct the
corresponding penetration cost from the budget, as each new container
is broken into. For a given assignment of costs and target values we obtain
a {\em security system}, and we discuss three types of them: 
{\em improved}, {\em good}, and {\em optimal}.  
We show that in general it is not possible to develop an
optimal security system for a given cyber-security model $M$. 
We define P- and C-models where the penetration costs and prizes, 
respectively, all have unit value.  We show that if
$T$ is a rooted tree such that any P- or C-model 
$M = (T,C,P)$ has an optimal security system, 
then $T$ is one of the following types:
(i) a rooted path, 
(ii) a rooted star, 
(iii) a rooted 3-caterpillar, or 
(iv) a rooted 4-spider. Conversely, if $T$ is one of these four types of
trees, then we show that any P- or C-model $M = (T,C,P)$ does have an
optimal security system\@. Finally, we study a duality between 
P- and C-models that allows us to translate results for P-models into 
corresponding results for C-models and vice versa. The results
obtained give us some mathematical insights into
how layered-security defenses should be organized.
\end{abstract}


\section{Introduction}
\label{sec:introduction}

According to~\cite{GoGulf}, the global cyber-security market cost
in~2017 is expected to top 120 billion US dollars.  This site also
reports that there are 18 victims of a cyber crime every single
second!  Other sources report similarly alarming and worsening
statistics.  There is agreement that the number of cyber attacks is
increasing rapidly, and the consequences of such attacks are greater
than ever on economics, national security, and personal data.  Threats
come from nation states with advanced cyber warfare commands, nation
states having less technical capabilities but intent on doing harm,
ideologically motivated groups of hackers or extremists,
profit-seeking criminals, and others.  Building on the work done
in~\cite{AgGrKa}, in this paper we study a layered-security model and
strategies for assigning {\em penetration costs\/} and 
{\em target-acquisition values\/} so as to minimize the amount of damage
an attacker can do to a system.  That is, we examine 
{\em security systems\/}. The approach we take here is to assign
weights to the vertices and edges of a tree in order to build a cyber
defense that minimizes the amount of prize an attacker can accumulate
given a limited budget.  To the best of our knowledge this approach is
new in that the usual approach is to consider a particular weighted
tree as input.  For example,
in~\cite{Houssaine-etal,AgGrKa,coene:trees,LNCS-bounded,DO-main,IPL-improved}
the authors consider finding 
{\em weight-constrained, maximum-density subtrees\/} 
and similar structures given a fixed weighting of a tree
as part of the input.  In these cases weights are specified on both
vertices and edges.

There has been some research on {\em network fortification\/} and
problems related to that topic.  For example, in~\cite{Uhan}
stochastic linear programming games are studied and it is demonstrated
how these can, among other things, model certain network
fortifications. In~\cite{Wood} the problem of network interdiction is
studied---how to minimize the maximum amount of flow an
adversary/enemy can push through a given network from a source $s$ to
a sink $t$.  There each edge/arc is provided with a fixed integer
capacity and an integer resource (required to delete the
edge/arc). This is a variation of the classical
Max-Flow-Min-Cut Theorem.  Although interesting in their own way,
neither of these papers or related papers that we have found in the
literature address directly what we study in this paper.

In~\cite{AgGrKa} the authors posed the following question: Can one
mathematically prove that the intuition of storing high-value targets
deeper in the system and having higher penetration costs on the
outer-most layers of the system results in the best or at least good
security?  In this paper we answer this question and obtain more general
and specific results.  We define three types of security systems: 
{\em improved}, {\em good}, and {\em optimal}. We show that not all 
{\em cyber-security models\/} admit
optimal security systems, but prove that paths and stars do.  
We define and study {\em P-} and {\em C-models\/} where 
all of the penetration costs, and the prizes, are set to one,
respectively. We classify the
types of trees that have optimal security systems for both P- and C-models.  
We then discuss a duality between P- and C-models, which provides a 
dictionary to translate results for P-models into corresponding results
for C-models, and vice versa.

To build secure systems requires first principles of security.  ``In
other words, we need a {\em science of cyber-security\/} that puts the
construction of secure systems onto a firm foundation by giving
developers a body of laws for predicting the consequences of design
and implementation choices''~\cite{schneider}.  To this end Schneider
called for more models and abstractions to study cyber
security~\cite{schneider}.  This paper is a step in that direction.
We hope that others will build on this work to develop even better and
more-realistic models, overcome the shortcomings of our model, as well
as develop additional foundational results.

The outline of this article is as follows.  In section~\ref{sec:model}
we present the rationale for our layered-security model. In 
Section~\ref{sec:csdm} we define
the framework for security systems and present the definitions of 
improved, good, and optimal security systems, 
and state some related observations and examples.
In section~\ref{sec:optimal} we explore optimal security systems and prove 
that they do not always exist, but show they exist if and only if
$T$ is either a path rooted at a leaf, or a star rooted at
its center vertex.
In section~\ref{sec:CP-gen-obs} we define P- and C-models and show
that any cyber-security model $M = (T,C,P)$ is equivalent to 
both a P-model $M'$ and a C-model $M''$. We further show that if 
$T$ is a rooted tree such
that any P- or C-model $M$ has an optimal security system,
then $T$ is one of the following four types: 
(i) a rooted path, 
(ii) a rooted star, 
(iii) a rooted 3-caterpillar, or 
(iv) a rooted 4-spider.
In section~\ref{sec:optimal-P-trees} we prove that if $T$ is one of
the four types of rooted trees mentioned above, then any P-model 
does indeed have an optimal security system\@.  
In section~\ref{sec:duality} we define a duality between equivalence
classes of P-models and equivalence classes of C-models that serves
as a dictionary allowing us to obtain equivalent results for C-models
from those of the P-models that were obtained in 
section~\ref{sec:optimal-P-trees}. In particular, we obtain 
Theorem~\ref{thm:C-P-main} that completely classifies which P-
and which C-models have optimal security systems. Conclusions
and open problems are discussed in section~\ref{sec:summary}.

\section{Rationale for Our Layered-Security Model}
\label{sec:model}

In defining our layered-security model to study defensive cyber
security, we need to strike a balance between simplicity and utility.
If the model is too simple, it will not be useful to provide insight
into real situations; if the model is too complex, it will be too
cumbersome to apply, and we may get bogged down in too many details.
The model described in this paper is a step toward gaining a better
understanding of a broad range of security systems in a 
graph-theoretical setting for a layered-security model.

Many systems contain layered security or what is commonly referred to
as {\em defense-in-depth}, where valuable assets are hidden behind
many different layers or secured in numerous ways.  For example, a
{\em host-based defense\/} might layer security by using tools such as
signature-based vendor anti-virus software, host-based systems
security, host-based intrusion-prevention systems, host-based
firewalls, encryption, and restriction policies, whereas a 
{\em network-based defense\/} might provide defense-in-depth by using
items such as web proxies, intrusion-prevention systems, firewalls,
router-access control lists, encryption, and
filters~\cite{johnston:lafever}.  To break into such a system and
steal a valuable asset requires that several levels of security be
penetrated, and, of course, there is an associated cost to break into
each level, for example, money spent, time used, or the punishment
served for getting caught.

Our model focuses on the layered aspect of security and is intended to
capture the notion that there is a cost associated with penetrating
each additional level of a system and that attackers have finite
resources to utilize in a cyber attack.  Defenders have the ability to
secure targets using defense mechanisms of various strengths and to
secure targets in desired locations and levels.  We assume that the
structure where targets will be stored; that is, the container
nestings; is given as part of the input in the form of a rooted tree.
In this way we can study {\em all\/} possible structures at a single
time, as they can be captured in the definition of our problems.  This
methodology is as opposed to having the defender actually construct a
separate defense structure for each input.

For any specific instance of a problem, a defender of a system will
obviously consider the exact details of that system and design a
layered-security approach to fit one's actual system.  Similarly, a
traveling salesman will be concerned about constructing a tour of 
{\em his\/} particular cities, not a tour of any arbitrary set of cities
with any arbitrary set of costs between pairs of cities.
Nevertheless, researchers have found it extremely helpful to consider
a general framework in which to study the 
{\sc Traveling Salesman Problem}.  And, in studying the general 
problem, insights have been
gained into {\em all\/} instances of the problem.  Thus, we believe it
is worthwhile to consider having a fixed structure as part of our
input, and this approach is not significantly different from that used
in complexity theory to study
problems~\cite{garey:johnson,greenlaw:etal}.

In this paper we focus on a static defense.  We pose as an open
problem the question of how to create a defense and an attack strategy
if the defender is allowed to move targets around dynamically or
redistribute a portion of a prize.  We also consider the total prize
as the sum of the individual values of the targets collected although
one could imagine using other or more-complex functions of the target
values to quantify the damage done by an attacker.  Our defensive
posture is formed by assigning to the edges and vertices of the rooted 
tree in question the input-provided penetration costs and
target-acquisition values, respectively.  We formalize the model, the
notion of a security system, and the concept of a system attack in the
next section.

\section{Cyber-Security Model and Security Systems}
\label{sec:csdm}

Let $\nats = \{1,2,3,\ldots\}$, $\rats$ be the rational numbers, and
$\prats$ be the non-negative rational numbers.
\begin{definition}
\label{def:GO-tree}
A {\em cyber-security model (CSM)\/} $M$ is given by a
three-tuple $M = (T, C, P)$, where $T$ is a directed tree rooted at $r$ having
$n \in \nats$ non-root vertices, $C$ is a multiset of 
{\em penetration costs\/} $c_1, \ldots, c_n \in {\prats}$, and $P$
is a multiset of {\em target-acquisition-values\/} 
(or {\em prizes\/} for short) $p_1, \ldots, p_n \in {\prats}$.
\end{definition}
Throughout $V(T) = \{r, u_1,\ldots, u_n\}$, where
$r$ is the designated root that indicates the start of a 
{\em system attack}, and $E(T) = \{e_1,\ldots, e_n\}$ denotes the set of 
edges of $T$, where our labeling is such that $u_i$ is always the
head of the edge $e_i$. 
The prize at the root is set to 0. The penetration costs model
the expense for breaking through a layer of security, and the
target-acquisition-values model the amount of prize one acquires for
breaking through a given layer and exposing a target.  The penetration
costs will be weights that are assigned to edges in the tree, and the
target-acquisition-values, or the prizes, are weights that will be assigned to
vertices in the tree.

Sometimes we do not distinguish a target from its acquisition
value/prize/reward nor a container, which is a layer of security, from
its penetration cost.  Note that one can think of each edge in the
rooted tree as another container, and as one goes down a path in the
tree, as penetrating additional layers of security.  We can assume
that the number of containers and targets is the same.  Since if we
have a container housing another container (and nothing else), we can
just look at this ``double'' container as a single container of
penetration cost equal to the sum of the two nested ones.  Also, if a
container includes many prizes, we can just lump them all into a
single prize, which is the sum of them all.

Recall that in a rooted tree $T$, each non-root vertex $u \in V(T)$ has
exactly one parent, and that we assume the edges of $T$
are directed naturally away from the root $r$ in such a way that each
non-root vertex has an in-degree of one. The root is located at 
{\em level 0\/} of the tree. {\em Level one\/} of the tree consists of the
children of the root, and, in general, {\em level i\/} of the tree
consists of the children of those vertices at level $i - 1$ for $i\geq 1$.  
We next present some key definitions about a CSM
that will allow us to study questions about {\em security systems}.
\begin{definition}
\label{def:GO-defensive}
A {\em security system (SS)\/} with respect to a cyber-security model
$M = (T, C, P)$ is given by two bijections $c : E(T) \rightarrow C$
and $p : V(T)\setminus\{r\}\rightarrow P$.  
We denote the security system by $(T, c, p)$.

A {\em system attack (SA)\/} in a security system $(T, c, p)$ is
given by a subtree $\tau$ of $T$ that contains the root $r$ of $T$.

\begin{itemize}
\item
The {\em cost\/} of a system attack $\tau$ with respect to a
security system $(T, c, p)$ is defined by
\[
\cost(\tau,c, p) = \sum_{e\in E(\tau)}c(e).
\]
\item
The {\em prize\/} of a system attack $\tau$ with respect to a
security system $(T, c, p)$ is defined by 
\[
\pr(\tau,c, p) = \sum_{u\in V(\tau)}p(u).
\]
\item
For a given {\em budget\/} $B \in \prats$ the 
maximum prize $\maxp(B,c,p)$ with respect to $B$ is defined by
\[
\maxp(B,c,p) := 
\max\{ \pr(\tau,c,p) : \mbox{for all system attacks } \tau\subseteq T,
\mbox{ where } \cost(\tau,c,p) \leq B \}.
\]
A system attack $\tau$ whose prize is a maximum with respect to a given
budget is called an {\em optimal attack}.
\end{itemize}
\end{definition}
The bijection $c$ in Definition~\ref{def:GO-defensive} specifies how
difficult it is to break into the various containers, and the
bijection $p$ specifies the prize associated with a given container.
Note that for any SS $(T, c, p)$ we have 
$\cost(r,c,p) = 0 \leq B \in \prats$.  
When $T = (\{ r\}, \varnothing )$, then $\maxp(B,c,p) = 0$
for any $B \in \prats$.  When two bijections are
given specifying a SS, we call the resulting weighted tree a 
{\em configuration of the CSM}\@.  Any configuration represents a
defensive posture and hence the name security system.  Note that the
CSM can be used to model any general security system and not just
cyber-security systems.  We are interested in configurations that make
it difficult for an attacker to accumulate a large prize.  It is
natural to ask if a given defensive stance can be improved.  Next we
introduce the notion of an {\em improved security system\/} that will
help us to address this question.
\begin{definition}
\label{def:improvedDefense}
Given a CSM $M = (T, C, P)$ and a SS $(T, c, p)$, an 
{\em improved security system (improved SS) with respect to $(T, c, p)$\/} 
is a SS $(T, c', p'$) such that for any
budget $B \in \prats$ we have $\maxp(B,c',p') \leq \maxp(B,c,p)$, 
and there exists some budget $B' \in \prats$ such that
$\maxp(B',c',p') < \maxp(B',c, p)$.
\end{definition}
Definition~\ref{def:improvedDefense} captures the idea of a better
placement of prizes and/or penetration costs so that an attacker
cannot do as much damage.  That is, in an improved SS one can never
acquire a larger overall maximum prize with respect to any budget $B$;
and furthermore, there must be at least one particular budget where
the attacker actually does worse.  Notice that there can be an
improved SS $(T, c', p')$, where for some budget $B \in \prats$,
there is a SA $\tau$ whose cost is less than or equal to $B$ for both
SSs such that $\pr(\tau,c', p') > \pr(\tau,c, p)$.  In this case an
attacker obtains a larger prize in the improved SS; and, of course,
this situation is undesirable and means a weaker defense against this
specific attack.  We, however, are interested in improved SSs with
respect to a given budget rather than a particular SA\@.  Since we
have exactly $n$ penetration costs and $n$ prizes to assign, it is
difficult to imagine an improved SS for all but the most-restricted
trees in which all SAs would be improved in the sense just described.
Next, we formalize the notion of an {\em optimal security system\/}.
\begin{definition}
\label{def:optimaldefensivestrategy}
Let $M = (T,C,P)$ be a given CSM.
(i) For a budget $B\in\prats$, a SS $(T,c,p)$ 
is {\em optimal w.r.t.~$B$} if there is no other SS 
$(T, c', p')$ for $M$ such that $\maxp(B,c',p') < \maxp(B,c,p)$.
(ii) $(T,c,p)$ is {\em optimal} if it is optimal w.r.t.~any budget 
$B\in\prats$.
\end{definition}
Notice that an optimal SS is not necessarily the best possible.  We
could define a {\em critically optimal security system\/} to be one
where for every single SA the SS was at least as good as all others
and for at least one better.  And, in a different context, these SSs
might be interesting.  However, in light of
Theorem~\ref{thm:nooptimal} in the following section, 
which shows that even an optimal SS may
not exist for a given CSM, we do not pursue critically optimal SSs
further in this paper. By Definitions~\ref{def:improvedDefense}
and~\ref{def:optimaldefensivestrategy} we clearly have the following.
\begin{observation}
\label{obs:optimalcannotbeimproved}
A SS $(T, c, p)$ for a CSM $M = (T,C,P)$ is optimal
if and only if no improved SS for $(T, c, p)$ exists.
\end{observation}

We next introduce the concept of two closely-related configurations of
a CSM, and this notion will give us a way to relate SSs.  
\begin{definition}
\label{def:configurationNeighbor}
Given a CSM $M = (T, C, P)$, the two configurations
$(T, c, p)$, and $(T, c', p')$ are said to be {\em neighbors\/} if
\begin{enumerate}
  \item there exists an edge $(u,v)\in E(T)$ such that
\begin{eqnarray*}
p'(v) & = & p(u)\\
p'(u) & = & p(v)\\
p'(w) & = & p(w), \mbox{ otherwise, or}
\end{eqnarray*}
  \item there exist two edges $(u, v), (v, w) \in E(T)$ such that
\begin{eqnarray*}
c'((u, v)) & = & c((v,w))\\
c'((v,w)) & = & c((u,v))\\
c'((x,y)) & = & c((x,y)), \mbox{ otherwise.}
\end{eqnarray*}
\end{enumerate}
\end{definition}
The notion of neighboring configurations will be useful in developing
algorithms for finding {\em good security systems}, which we define
next.
\begin{definition}
\label{def:goodstrategy}
A {\em good security system (good SS)\/} is a SS
$(T,c,p)$ such that no neighboring configuration results in an improved
security system.
\end{definition}
Given a SS $(T,c,p)$ for a CSM $M$, a natural question
to pose is whether a local change to the SS can be made in order to 
strengthen the SS, that is, make the resulting SS improved.
In a practical setting one may not be able to redo the
security of an entire system, but instead may be able to make local
changes. 

Suppose $(u,v)\in E(T)$ where $p(u)\geq p(v)$,
and let $p'$ be the prize assignment obtained from $p$ by
swapping the prizes on $u$ and $v$, that is $p'(u)=p(v)$,
$p'(v) = p(u)$, and $p'(w)=p(w)$ otherwise. If now
$\tau$ is any SA, then $\pr(\tau,c,p') = \pr(\tau,c,p)$
if either both $u,v\in V(\tau)$ or neither $u$ nor $v$ are
vertices of $\tau$, or $\pr(\tau,c,p') \leq \pr(\tau,c,p)$
if $u\in V(\tau)$ and $v\not\in V(\tau)$. In either case
$\pr(\tau,c,p')\leq \pr(\tau,c,p)$ and therefore we have
for any budget $B$ that
\begin{equation}
\label{eqn:p'leqp}
\maxp(B,c,p')\leq \maxp(B,c,p).
\end{equation}
Similarly, if $(u,v),(v,w)\in E(T)$ where
$c((u,v))\leq c((v,w))$, let $c'$ be the cost
assignment obtained from $c$ by swapping the costs
on the incident edges $(u,v)$ and $(v,w)$ and leave all the other
edge-costs unchanged, that is $c'((u,v))=c((v,w))$,
$c'((v,w))=c((u,v))$ and $c'(e)=c(e)$ otherwise.
If $\tau$ is a SA, then clearly we always have
$\pr(\tau,c',p) = \pr(\tau,c,p)$. Also, if either
both $(u,v), (v,w)\in E(\tau)$ or neither $(u,v)$ nor
$(v,w)$ are edges in $\tau$, then 
$\cost(\tau,c',p) = \cost(\tau,c,p)$, and if
$(u,v)\in E(\tau)$ and $(v,w)\not\in E(\tau)$,
then $\cost(\tau,c',p) \geq \cost(\tau,c,p)$. In either
case we have $\cost(\tau,c',p) \geq \cost(\tau,c,p)$.
Hence, if $B$ is any budget, then by mere definition
we have that
\begin{equation}
\label{eqn:c'leqc}
\maxp(B,c',p)\leq \maxp(B,c,p).
\end{equation}
By (\ref{eqn:p'leqp}) and (\ref{eqn:c'leqc}) we have the
following proposition.
\begin{proposition}
\label{prp:goodstrategy}
Let $M = (T, C, P)$ be a CSM. A SS
given by $(T, c, p)$ is a good SS if for all
$(u, v), (v, w) \in E$ we have $c((u,v)) \geq c((v,w))$ and for all 
non-root vertices $u, v \in V(T)$ with $(u,v) \in E(T)$ we have 
$p(u) \leq p(v)$.
\end{proposition}
Note that Proposition~\ref{prp:goodstrategy} says that on any root to leaf path
in $T$ the penetration costs occur in decreasing order and the
prizes occur in increasing order.

From any configuration resulting from a SS $(T, c, p)$ for a CSM,
Proposition~\ref{prp:goodstrategy} gives a natural $O(n^2)$ algorithm for
computing a good SS by repeatedly moving to improved neighboring
configurations until no more such neighboring configurations exist.
We can do better than this method by first sorting the values in $C$
and $P$ using $O(n \log n)$ time, and then conducting a breath-first
search of $T$ in $O(n)$ time.  We can then use the breath-first search
level numbers to define bijections $c$ and $p$ that meet the
conditions of a good SS\@.  We summarize in the following.
\begin{observation}
Given a CSM $M = (T, C, P)$, there is an 
$O(n\log n)$ algorithm for computing a good SS for $M$.
\end{observation}
If we could eliminate the sorting step, we would have a more efficient
algorithm for obtaining a good SS, or if we restricted ourselves to
inputs that could be sorted in $O(n)$ time.  Also, notice that a good
SS has the heap property, if we ignore the root.  However, in our case
we cannot ``chose'' the shape of the heap, but we must use the
structure that is given to us as part of our input.

Suppose that our SS $(T,c,p)$ for $M$ 
satisfies a {\em strict} inequality $p(u) > p(v)$ 
for some $(u,v)\in E(T)$, or that $c((u,v)) < c((v,w))$
for some incident edges $(u,v),(v,w)\in E(T)$. A natural
question is whether the prize and cost assignments $p'$ 
and $c'$ as in (\ref{eqn:p'leqp}) and (\ref{eqn:c'leqc}) 
will result in an improved SS as in 
Definition~\ref{def:improvedDefense}. In 
Example~\ref{exa:P-T_p(3)} we will see that that is
not the case.

{\sc Convention:} Let $T_p(\ell)$ denote the 
rooted tree whose underlying graph is a path on $2\ell + 1$ vertices 
$V(T_p(\ell)) = \{r,u_1,\ldots,u_{2\ell}\}$
and directed edges 
\[
E(T_p(\ell)) = \{ (r,u_1), (r,u_2), (u_1,u_3), (u_2,u_4), 
\ldots, (u_{2\ell-3},u_{2\ell-1}), (u_{2\ell-2},u_{2\ell})\}
\]
rooted at its center vertex. 
We label the edges by the same index
as their heads: $e_1 = (r,u_1)$, $e_2 = (r,u_2)$,...,  
$e_{2\ell-1} = (u_{2\ell-3},u_{2\ell-1})$, and 
$e_{2\ell} = (u_{2\ell-2},u_{2\ell})$.
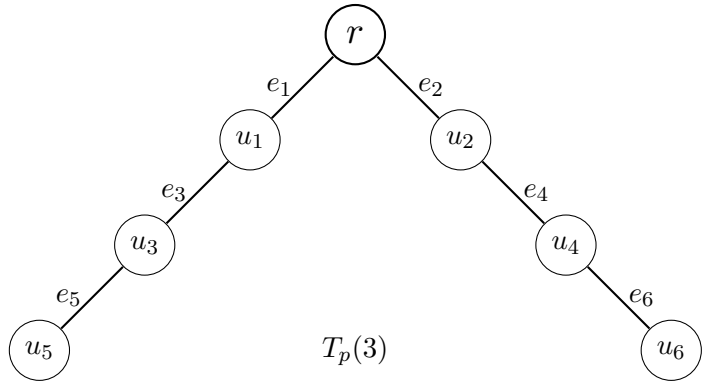
\begin{center}
\begin{tikzpicture}[scale=.7] 

\node[draw, thick, shape=circle,scale=1.3] (r) at (0,6) {$r$};
\node[draw, shape=circle,scale=1] (u1) at (-2,4) {$u_1$};
\node[draw, shape=circle,scale=1] (u2) at (2,4) {$u_2$};
\node[draw, shape=circle,scale=1] (u3) at (-4,2) {$u_3$};
\node[draw, shape=circle,scale=1] (u4) at (4,2) {$u_4$};
\node[draw, shape=circle,scale=1] (u5) at (-6,0) {$u_5$};
\node[draw, shape=circle,scale=1] (u6) at (6,0) {$u_6$};

\foreach \from/\to in {r/u1,r/u2,u1/u3,u2/u4,u3/u5,u4/u6}
\path[draw, thick] (\from) -- (\to);

\node[left] at (-1,5) {$e_1$};
\node[right] at (1,5) {$e_2$};
\node[left] at (-3,3) {$e_3$};
\node[right] at (3,3) {$e_4$};
\node[left] at (-5,1) {$e_5$};
\node[right] at (5,1) {$e_6$};

\node at (0,0) {$T_p(3)$};

\end{tikzpicture}
\end{center}
\begin{example}
\label{exa:P-T_p(3)}
\end{example}
\vspace{-2 mm} 
\noindent Let $(T_p(3),c,p)$ be a SS for a CSM $M$ where 
\[
c(e_1,e_2,e_3,e_4,e_5,e_6) := (1,1,1,1,1,2), \ \   
p(u_1,u_2,u_3,u_4,u_5,u_6) := (10,2,10,3,10,40),
\]
where the penetration costs and the prizes have
been simultaneously assigned in the obvious way.
We see that for any budget $B\in\prats$ we have
\[
\maxp(B,c,p) = \left\{
\begin{array}{ll}
10\lfloor B\rfloor & \mbox{ for } 0\leq B < 4, \\
10\lfloor B\rfloor + 5 & \mbox{ for } 4\leq B\leq 7, \\
75                     & \mbox{ for } 4 < B.
\end{array}
\right.
\]
If now $p'(u_1,u_2,u_3,u_4,u_5,u_6) = (10,3,10,2,10,40)$ is the prize assignment
obtained from $p$ by swapping the prizes on the neighboring vertices
$u_2$ and $u_4$, and $c'(e_1,e_2,e_3,e_4,e_5,e_6) = (1,1,1,2,1,1)$ be the 
edge-cost assignment obtained from $c$ be swapping the 
costs of the incident edges $e_4$ and $e_6$, then 
\[
\maxp(B,c,p') = \maxp(B,c',p) = \maxp(B,c,p),
\]
for any non-negative budget $B\in\prats$, showing that 
locally swapping either prize assignments on adjacent vertices,
or edge-costs on incident edges, does not necessarily improve the SS.

\vspace{3 mm}

In Theorem~\ref{thm:nooptimal} in section~\ref{sec:optimal}, 
we show that there are CSMs for which no optimal SS exists. 
In such cases obtaining a
locally optimal SS, as defined in Definition~\ref{def:goodstrategy},
may provide us with a reasonable defensive posture.  
In section~\ref{sec:optimal} we also show that optimal SSs 
exist for paths rooted at one of their leaves and for stars rooted at their
center vertices.

\section{Optimal Security Systems}
\label{sec:optimal}

One of the most natural and important questions to consider for a
given CSM $M$ is whether an optimal SS exists and if it does, what it
would look like.  Unfortunately, Theorem~\ref{thm:nooptimal} shows
that there are small and simple CSMs for which no optimal SS exists.  
Still we would like to
know for what CSMs optimal SSs do exist, and, if possible, have a way
to find these optimal SSs efficiently.  Corollary~\ref{cor:path-optimal}
and Theorem~\ref{thm:star} show that optimal SSs exist for CSMs 
$M = (T, C, P)$ when $T$ is a path or a star, respectively.  These 
theorems also
yield $O(n \log n)$ algorithms for producing optimal SSs in these
cases.  But, these results are not satisfying, as they are limited.
In sections~\ref{sec:CP-gen-obs}, \ref{sec:optimal-P-trees}, 
and~\ref{sec:duality} we
study {\em P-\/} and {\em C-models\/} and completely characterize the
types of trees that have optimal SSs.  

We begin with a lemma showing that all optimal SSs must have the
highest penetration costs assigned to the edges involving the root and
level-one vertices.
\begin{lemma}
\label{lem:level1}
Let $M = (T, C, P)$ be a CSM, where $T$ rooted at $r$
contains at least one non-root vertex.  Let $V_1 \subseteq T(V)$
denote the level-one vertices of $T$, and let $C_L$ be the multiset of
the largest $|V_1|$ values in $C$\@.  If an optimal $(T, c, p)$ SS for $M$, 
exists, then $c(e) \in C_L$ for $e \in \{ (r,v) \mid v \in V_1 \}$.
\end{lemma}
\begin{proof}
Suppose we have an optimal SS $(T, c, p)$ that does not meet the
conditions of the lemma.  Let $c_s \not\in C_L$ be the smallest
penetration cost assigned by $c$ to an edge between the root $r$ and a
vertex $v_s \in V_1$, that is, $c((r, v_s)) = c_s \leq c((r,v))$ for
all $v \in V_1 -\{ v_s \}$. Let $e_s = (r,v_s)$ and let $e_l$ be an
edge not between the root and a level-one vertex where $c(e_l) \in C_L$.
We know that such an edge exists because $(T, c, p)$ does not meet the
conditions of the lemma.
To show that $(T, c, p)$ cannot be an optimal SS, we define a SS 
$(T,c',p)$ by letting $c'(e_s) = c(e_l)$, $c'(e_l) = c(e_s)$,
and $c'(e) = c(e)$ otherwise. Notice that for the budget $B = c_s$, 
we have $\maxp(B,c,p) = p(v_s) > 0 = \maxp(B,c',p)$.  
This fact contradicts that $(T, c, p)$ is an optimal SS.
\end{proof}
If an optimal SS exists, Lemma~\ref{lem:level1} tells us something
about its form.  In the next theorem we show that there are CSMs for
which no optimal SS exists.
\begin{theorem}
\label{thm:nooptimal}
There is a CSM $M = (T, C, P)$ for which no optimal
security system exists.
\end{theorem}
\begin{proof}
Consider $M = (T,\{1,2,3\},\{1,2,3\})$, 
where $T$ is the tree given by $V(T) = \{r,u_1,u_2,u_3\}$ and 
$E(T) = \{e_1,e_2,e_3\}$ where $e_1 = (r, u_1)$, $e_2 = (r, u_2)$,  
and $e_3 = (u_1, u_3)$. By Lemma~\ref{lem:level1} we know
that an optimal SS $(T,c,p)$ has 
$c(e_3) = 1$, and we can further assume that $p(u_3) = 3$. 
By considering the budget of $B=2$,
we can also assume the prize of the head of the edge of cost $2$
to by $1$. Therefore, we have only two possible
optimal SSs for $M$: $(T,c,p)$ with $c(e_1,e_2,e_3) = (3,2,1)$
and $p(u_1,u_2,u_3) = (2,1,3)$, or $(T,c',p')$ with
$c'(e_1,e_2,e_3) = (2,3,1)$ and $p'(u_1,u_2,u_3) = (1,2,3)$. 
Since $\maxp(3,c,p) = 2$ and $\maxp(3,c',p') = 4$, we see
that $(T,c',p')$ is not optimal, and since 
$\maxp(4,c,p) = 5$ and $\maxp(4,c',p') = 4$, we see that $(T,c,p)$
is not optimal either. Hence, no optimal SS for $M$ exists.
\end{proof}
Although Theorem~\ref{thm:nooptimal} showed that there are CSMs for
which no optimal SS exists, we are interested in finding out for which
trees $T$ optimal SSs do exist.  We should point out that the values
of the weights in $C$ and $P$ also play an important role in whether
or not an optimal SS exists for a given tree.  In the next theorem we
show that an optimal SS exists for CSMs in which the tree in the model
is a path, and this result is independent of the values of the weights
in $C$ and $P$\@.  

Consider a CSM $M = (T,C,M)$ where $T$ is a path rooted at a leaf, so
\begin{equation}
\label{eqn:rooted-path}
V(T) = \{u_0,u_1,\ldots, u_n\}, \ \ E(T) = \{e_1,\ldots,e_n\}, 
\end{equation} 
where $u_0 = r$ and $e_i = (u_{i-1},u_i)$,
for each $i\in\{1,\ldots,n\}$. For a SS $(T,c,p)$ for $M$, then
for convenience let $p_i = p(u_i)$ and $c_i = c(e_i)$ for each $i$.
If we have $p_i\leq p_{i+1}$ and $c_i\geq c_{i+1}$ for each 
$i\in\{1,\ldots,n-1\}$ (so the prizes are ordered increasingly
and the edge-costs decreasingly as we go down the path from the root),
then by Proposition~\ref {prp:goodstrategy} the SS $(T,c,p)$ is 
a good SS as in Definition~\ref{def:goodstrategy}. But, we can
say slightly more here when $T$ is a path, in terms of obtaining
an improved SS as in Definition~\ref{def:improvedDefense}.
\begin{lemma}
\label{lmm:path-improved}
Let $M = (T,C,M)$ be a CSM where $T$ is a path with its vertices
and edges labeled as in (\ref{eqn:rooted-path}).

(i) If $(T,c,p)$ is a SS for $M$ and there is an $i$ with 
$p_i > p_{i+1}$ and $c_{i+1} > 0$, then the SS $(T,c,p')$ where
$p'$ is obtained by swapping the prizes on $u_i$ and $u_{i+1}$
is an improved SS.

(ii) If $(T,c,p)$ is a SS for $M$ and there is an $i$ with 
$c_i < c_{i+1}$, then the SS $(T,c',p)$ where $c'$ is obtained
by swapping the edges costs on $e_i$ and $e_{i+1}$ is an 
improved SS.
\end{lemma}
\begin{proof}
By Proposition~\ref {prp:goodstrategy} we only need to show
(i) there is a budget $B'$ such that $\maxp(B',c,p') < \maxp(B',c,p)$
and (ii) a budget $B''$ such that $\maxp(B'',c',p) < \maxp(B'',c,p)$.
For each $j$ let $\tau_{j} = T[e_1,\ldots,e_{j}]$ be the
rooted sub-path of $T$ that contains the first $j$ edges of $T$.

For $B' = c_1 +\cdots + c_i$ we clearly have 
\[
\maxp(B',c,p') 
= \pr(\tau_i,c,p') 
= p_1 + \cdots + p_{i-1} + p_{i+1} 
< p_1 + \cdots + p_i
= \pr(\tau_i,c,p) 
= \maxp(B',c,p), 
\]
showing that $(T,c,p')$ is an improved SS for $M$.

Likewise, we have
\[
\maxp(B',c',p) 
= \pr(\tau_{i-1},c',p) 
= p_1 + \cdots + p_{i-1}
< p_1 + \cdots + p_i
=  \pr(\tau_i,c,p) 
= \maxp(B',c,p),
\]
showing that $(T,c',p)$ is also an improved SS for $M$.
\end{proof}
Given any SS $(T,c,p)$ for $M$ as in Lemma~\ref{lmm:path-improved}
when $T$ is a rooted path, by bubble sorting the prizes and the edge
costs increasingly and decreasingly respectively, as we go down the 
path $T$ from the root, we obtain by Lemma~\ref{lmm:path-improved}
a SS $(T,c',p')$ such that for any budget $B$ we have
$\maxp(B,c',p')\leq \maxp(B,c,p)$. We therefore have
the following corollary.
\begin{corollary}
\label{cor:path-optimal}
If $M = (T,C,M)$ is a CSM where $T$ is a rooted path with its vertices
and edges labeled as in (\ref{eqn:rooted-path}),
then there is an optimal SS for $M$, and it is given 
by assigning the penetration costs to the edges 
and the prizes to the vertices in a decreasing order
and increasing order respectively from the root.
\end{corollary}

We now show that an optimal SS exists for $M = (T,C,P)$ when $T$ is a star.
Let $T$ be a star with root $r$ and non-root vertices $u_1,\ldots,n_n$
and edges $e_i = (r,u_i)$ for $i=1,\ldots,n$. Suppose the costs
and prizes are given by $C = \{c_1, \ldots, c_n \}$ and 
$P = \{p_1, \ldots, p_n \}$. When considering an arbitrary security system 
$(T,c,p)$ where $c(u_i) = c_i$ and $p(e_i) = p_i$ for each $i$,
we can without loss of generality assume the edge-costs to be 
in an increasing order $c_1\leq\cdots\leq c_n$. 
\begin{lemma}
\label{lmm:star}
Suppose $T$ is a star and $(T,c,p)$ is a SS as above. 
If $p'$ is another prize assignment obtained from $p$ by swapping 
the prizes $p_i$ and $p_j$ where $i<j$ and $p_i\leq p_j$, 
then for any budget $B$ we have $\maxp(B,c,p)\leq \maxp(B,c,p')$.
\end{lemma}
\begin{proof}
Let $B$ be a given budget and ${\tau}\subseteq T$ an optimal attack
with respect to $p$, so $\pr({\tau},c,p) = \maxp(B,c,p)$. 
We consider the following cases.

{\sc Case one:} If both of $u_i$ and $u_j$ are in 
${\tau}$, or neither of them are,
then $\maxp(B,c,p) = \pr({\tau},c,p) = \pr({\tau},c,p')\leq \maxp(B,c,p')$.

{\sc Case two:} If $u_i\in V({\tau})$ and $u_j\not\in V({\tau})$, then
$\maxp(B,c,p) = \pr({\tau},c,p) \leq \pr({\tau},c,p) - p_i + p_j
= \pr({\tau},c,p') \leq \maxp(B,c,p')$.

{\sc Case three:} If $u_i\not\in V({\tau})$ and $u_j\in V({\tau})$, then
${\tau'} = ({\tau}- u_j)\cup u_i$ is a rooted subtree of $T$ with
$c({\tau'}) = c({\tau}) - c_j + c_i \leq B$ and is therefore 
within the budget $B$.
Hence, $\maxp(B,c,p) = \pr({\tau},c,p) = \pr({\tau'},c,p')\leq \maxp(B,c,p')$.

Therefore, in all cases we have $\maxp(p,c,B)\leq \maxp(p',c,B)$.
\end{proof} 
Since any permutation is a composition of transpositions, we have the
following theorem as a corollary.
\begin{theorem}
\label{thm:star}
Let $M = (T, C, P)$ be a CSM where $T$ is a star rooted at its 
center vertex. Then there is an optimal SS for $M$, and it is given
by assigning the prizes to the vertices in the same increasing order 
as the costs are assigned increasingly to the corresponding edges.
\end{theorem}
For rooted trees on $n$ non-root vertices, 
Corollary~\ref{cor:path-optimal} and Theorem~\ref{thm:star} give
rise to natural sorting-based $O(n \log n)$ algorithms for computing
optimal SSs. Notice that in an optimal SS in a general tree,
the smallest prize overall
must be assigned to a level-one vertex $u$ which has the largest
penetration cost assigned to its corresponding edge, $(r, u)$, to the
root.  And, furthermore, we cannot say more than this statement for
arbitrary trees as the next assignment of a prize will depend on the
relative values of the penetration costs, prizes, and structure of the
tree. In view of the fact that optimal SSs do not
exist, except for paths and stars as we will see shortly in 
Observation~\ref{obs:path-and-star}, we turn our attention to
restricted CSMs and classify them with respect to optimal SSs.

\section{Specific Security Systems, P-Models, and C-Models}
\label{sec:CP-gen-obs} 

In this section we extend CSMs to include penetration costs and prizes
of value zero.  For a CSM $M = (T, C, P)$ with no optimal SS and a
rooted super-tree $T^{\dagger}$ of which $T$ is a rooted subtree, we can
always assign the prize of zero to the nodes in $V(T^{\dagger})\setminus V(T)$
and likewise the penetration cost of zero to the edges in
$E(T^{\dagger})\setminus E(T)$, thereby obtaining a CSM 
$M^{\dagger} = (T^{\dagger}, C^{\dagger},
E^{\dagger})$ that also has no optimal SS\@. Hence, by the example provided in
the proof of Theorem~\ref{thm:nooptimal}, we have the following
observation.
\begin{observation}
\label{obs:path-and-star}
If $T$ is a rooted tree, such that for any multisets $C$ and $P$ of
penetration costs and prizes, respectively, the CSM 
$M = (T,C,P)$ has an optimal SS, then $T$ is either a path rooted at
one of its leaves, or a star rooted at its center vertex.
\end{observation}
In light of Observation~\ref{obs:path-and-star}, we seek some natural
restrictions on our CSM $M$ that will guarantee it having an optimal
SS.  Since both the penetration costs and the prizes of $M = (T,C,P)$
take values in $\prats$ we can, by an appropriate
scaling, obtain an equivalent CSM where both the costs and prizes take
values in $\nats \cup \{ 0 \}$, that is, we may assume $c(e)\in\nats
\cup \{ 0 \}$ and $p(u)\in\nats \cup \{ 0 \}$ for every $e\in
E(T)$ and $u\in V(T)$, respectively.

First, we consider the restriction on a CSM $M = (T,C,P)$ where $C$
consists of a single penetration-cost value, that is, 
$C = \{1,1,\ldots, 1\}$ consists of $n$ copies of the unit penetration 
cost one. From a realistic point of view, this assumption seems to be
reasonable; many computer networks consist of computers with similar
password/encryption security systems on each computer (that is, the
penetration cost is the same for all of the computers), whereas the
computers might store data of vastly distinct values (that is, the
prizes are distinct).

{\sc Convention:} In what follows, it will be convenient
to denote the multiset containing $n$ (or an arbitrary number of) 
copies of $1$ by ${{I}}$. In a similar way, we will denote by $\mathbf{1}$
the map that maps each element of the appropriate domain to $1$.
As the domain of $\mathbf{1}$ should be self-evident each time, 
there should be no ambiguity about it each time.
\begin{definition}
\label{def:P-model}
A {\em P-model\/} is a CSM $M = (T,{{I}},P)$ where $T$
has $n$ non-root vertices and where ${{I}}$ is constant, consisting 
of $n$ copies of the unit penetration cost.
\end{definition}
Consider a SS $(T,c,p)$ of a CSM $M = (T, C, P)$. We can obtain an
{\em equivalent\/} SS $(T',\mathbf{1},p')$ of a P-model 
$M' = (T', {{I}}, P')$ in the
following way: for each edge $e = (u,v) \in E(T)$ with penetration
cost $c(e)=k\in\nats$ and prizes $p(u), p(v) \in\nats$ of its head and
tail, respectively, replace the 1-path $(u,e,v)$ with a directed path
of new vertices and edges $(u,e_1,u_1,e_2,u_2,\ldots, u_{k-1},e_k,v)$
of length $k$.  We extend the penetration cost and prize functions 
by adding zero-prize vertices where needed, that is, ${\mathbf{1}}(f) = 1$ 
for each $f\in E(T')$, and we let
\[
p'(u) = p(u), \ \ p'(v) = p(v), \mbox{ and }
p'(u_1) = p'(u_2) = \cdots = p'(u_{k-1}) = 0.
\]
In this way we obtain a SS $(T',c',p')$ of a P-model $M' = (T',I,P')$.
We view the vertices $V(T)$ of positive prize as a subset of
$V(T')$ (namely, those vertices of $T'$ with positive
prize).\footnote{Note that there are some redundant definitions on the
prizes of the vertices when considering incident edges, but the
assignments do agree, as they have the same prize values as in $T$.}

Recall that $T$ is a {\em rooted contraction\/} of $T'$ if $T$ is
obtained from $T'$ by a sequence of simple contractions of edges, and
where any vertex contracted into the root remains the root.  With
rooted trees, this means precisely that $T$ is a rooted 
{\em minor\/} of $T'$~\cite[p.~54]{Geir-Ray}.
\begin{proposition}
\label{prp:P-model}
Any SS $(T,c,p)$ of a CSM $M = (T,C,P)$ is
equivalent to a SS $(T',\mathbf{1},p')$ of a P-model 
$M' = (T',I,P')$ where
(i) $T$ is rooted minor of $T'$, and 
(ii) $p'(u) = p(u)$ for each $u\in V(T)\subseteq V(T')$,
and $p'(u) = 0$, otherwise.
\end{proposition}
\begin{proof}
(Sketch) Given a budget $B\in\prats$, clearly any optimal attack
$\tau$ on a SS $(T,c,p)$ having $\pr(\tau,c, p) = \maxp(B,c,p)$ has an
equivalent attack $\tau'$ on a SS $(T',\mathbf{1},p')$ of the same cost
$\cost(\tau',\mathbf{1}, p') = \cost(\tau,c, p)$ and hence within the
budget $B$, where $\tau'$ is the smallest subtree of $T'$ that
contains all of the vertices of $\tau$. By construction, we also
have that $\pr(\tau',\mathbf{1}, p') = \pr(\tau,c, p) = \maxp(B,c,p)$ 
since all of the vertices from $\tau$ are in $\tau'$ and have the 
same prize there, and the other vertices in $\tau'$ have prize zero. 
This shows that $\maxp(B,c,p)\leq \maxp(B,\mathbf{1},p')$.

Conversely, an optimal attack $\tau'$ on 
$(T',\mathbf{1},p')$ with $\pr(\tau',\mathbf{1},p') = \maxp(B,\mathbf{1},p')$ 
yields an attack $\tau$ on $(T,c,p)$ by
letting $\tau$ be the subtree of $T$ induced by the vertices
$V(\tau')\cap V(T)$. In this way 
$\pr(\tau,c, p) = \pr(\tau',\mathbf{1}, p')$ and 
$\cost(\tau,c, p) \leq \cost(\tau',\mathbf{1}, p')$, since some
of the vertices of $\tau'$ might have zero prize, as they are not in
$\tau$. By definition of $\maxp(\cdot)$ we have that
$\maxp(B,\mathbf{1},p')\leq \maxp(B,c,p)$. Hence, the SS $(T,c,p)$ and
$(T',\mathbf{1},p')$ are equivalent.
\end{proof}
Secondly, and dually, we can restrict our attention to the case where
the multiset of prizes $P$ consists of a single unit prize value, 
so $P = {{I}} = \{1,1,\ldots,1\}$ 
consists of $n$ copies of the unit prize.
\begin{definition}
\label{def:C-model}
A {\em C-model\/} is a CSM $M = (T,C,{{I}})$, 
where $T$ has $n$ non-root vertices and where ${{I}}$ is constant,
consisting of $n$ copies of the unit prize.
\end{definition}
As before, consider a SS $(T,c,p)$ of a CSM $M = (T, C, P)$.  We can
obtain an equivalent SS $(T'',c'',\mathbf{1})$ of a C-model 
$M'' = (T'',C'',I)$ in the following way: for each edge 
$e = (u,v) \in E(T)$ with
penetration cost $c(e)=k\in\nats$ and prizes $p(u), p(v) \in\nats$ of
its head and tail, respectively, replace the 1-path $(u,e,v)$ with a
directed path of new vertices and edges $(u,e,u_1,e_1,u_2,\ldots,
u_{k-1},e_{k-1},v)$ of length $k$.  We extend the penetration cost and
prize functions by adding zero-cost edges where needed,
that is, $\mathbf{1}(w) = 1$ for every $w\in V(T'')$, and we let
\[
c''(e) = c(e) \ \ \mbox{ and }
c''(e_1) = c''(e_2) = \cdots = c''(e_{k-1}) = 0.
\]
In this way we obtain a SS $(T'',c'',\mathbf{1})$ of a C-model 
$M'' = (T'',C'',I)$, where the multiset 
of prizes consists of a single unit prize
value ($\sum_{u\in V(T)\setminus\{r\}}p(u)$ copies of it).  We also
view the edges $E(T)$ of positive penetration cost as a subset of
$E(T'')$ (namely, those edges of $T''$ with positive penetration
cost).  We also have the following proposition that is dual to
Proposition~\ref{prp:P-model}.
\begin{proposition}
\label{prp:C-model}
Any SS $(T,c,p)$ of a CSM $M = (T,C,P)$ is
equivalent to a SS $(T'',c'',\mathbf{1})$ of a C-model 
$M'' = (T'', C'',I)$, where
(i) $T$ is rooted minor of $T''$, and 
(ii) $c''(e) = c(e)$ for each $e\in E(T)\subseteq E(T'')$,
and $c''(e) = 0$, otherwise.
\end{proposition}
\begin{proof}
(Sketch) Suppose we are given a budget $B\in\prats$ and an optimal
attack $\tau$ on a SS $(T,c,p)$ with 
$\pr(\tau,c, p) = \maxp(B,c,p)$.  
Here $(T'',c'',\mathbf{1})$ has an equivalent attack
$\tau''$, where $\tau''$ is the largest subtree of $T''$ that
contains all of the edges of $\tau$ and no other edges of $T$\@.
Note that $\cost(\tau'',c'', \mathbf{1}) = \cost(\tau,c, p)$ since all of
the additional edges of $\tau''$ that are not in $V(\tau)$ have zero
penetration cost, and so $\tau''$ is within the budget $B$\@.  Also,
by construction we have 
$\pr(\tau'',c'', \mathbf{1}) = \pr(\tau,c, p) = \maxp(B,c,p)$.  
This result shows that $\maxp(B,c,p)\leq \maxp(B,c'',\mathbf{1})$.

Conversely, consider an optimal attack $\tau''$ on $(T'',c'',\mathbf{1})$
with $\pr(\tau'',c'', \mathbf{1}) = \maxp(B,c'',\mathbf{1})$. 
By the optimality of
$\tau''$, every leaf of $\tau''$ is a tail of an edge of $T$, since
otherwise we can append that edge (of zero penetration cost), and
thereby obtain an attack with a prize strictly more than 
$\pr(\tau'',c'',\mathbf{1})$, a contradiction. 
The edges $E(\tau'')\cap E(T)$ induce
a subtree $\tau$ of $T$ of the same cost 
$\cost(\tau,c, p) = \cost(\tau'',c'', \mathbf{1})$; 
and moreover, $\tau''$ is, by its
optimality, the largest subtree of $T''$ that contains exactly all of the edges
of $\tau$, and so 
$\pr(\tau,c, p) = \pr(\tau'',c'', \mathbf{1}) = \maxp(B,c'',\mathbf{1})$. 
This result shows that 
$\maxp(B,c'',\mathbf{1})\leq \maxp(B,c,p)$. 
This proves that the SS $(T,c,p)$ and $(T'',c'',\mathbf{1})$
are equivalent.
\end{proof}

We now present some examples of both P- and C-models that will play
a pivotal role in our discussion to come.
\begin{definition}
\label{def:T2T3}
Let $T(2)$ denote the rooted tree given as follows:
\begin{eqnarray*}
V(T(2)) & = & \{r,u_1,u_2,u_3,u_4,u_5\}, \\
E(T(2)) & = & \{ (r,u_1), (r,u_2), (u_1,u_3), (u_2,u_4), (u_2,u_5)\}.
\end{eqnarray*}
Note that $T(2)$ has all of its non-root vertices on two non-zero levels.
Similarly, let $T(3)$ denote the rooted tree given as follows:
\begin{eqnarray*}
V(T(3)) & = & \{r,u_1,u_2,u_3,u_4\}, \\
E(T(3)) & = & \{ (r,u_1), (r,u_2), (u_2,u_3), (u_3,u_4)\}.
\end{eqnarray*}
Note that $T(3)$ has all of its vertices on three non-zero levels. 
\end{definition}

{\sc Convention:} For convenience we label the edges of both $T(2)$
and $T(3)$ with the same index as their heads:
\begin{eqnarray*}
T(2) & :  & e_1 = (r,u_1), \ e_2 = (r,u_2), \ e_3 = (u_1,u_3), \
  e_4 = (u_2,u_4), \ e_5 = (u_2,u_5). \\
T(3) & :  & e_1 = (r,u_1), \ e_2 = (r,u_2), \ e_3 = (u_1,u_3), \
  e_4 = (u_3,u_4).
\end{eqnarray*}

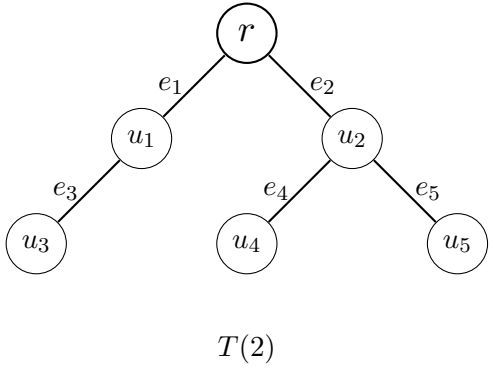
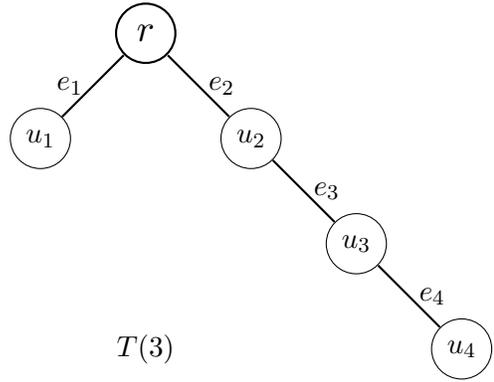
\begin{center}
\begin{tikzpicture}[scale=.7]

\node[draw, thick, shape=circle,scale=1.3] (r) at (-6,6) {$r$};
\node[draw, shape=circle,scale=1] (u1) at (-8,4) {$u_1$};
\node[draw, shape=circle,scale=1] (u2) at (-4,4) {$u_2$};
\node[draw, shape=circle,scale=1] (u3) at (-10,2) {$u_3$};
\node[draw, shape=circle,scale=1] (u4) at (-6,2) {$u_4$};
\node[draw, shape=circle,scale=1] (u5) at (-2,2) {$u_5$};

\foreach \from/\to in {r/u1,r/u2,u1/u3,u2/u4,u2/u5}
\path[draw, thick] (\from) -- (\to);

\node[left] at (-7,5) {$e_1$};
\node[right] at (-5,5) {$e_2$};
\node[left] at (-9,3) {$e_3$};
\node[left] at (-5,3) {$e_4$};
\node[right] at (-3,3) {$e_5$};

\node at (-6,0) {$T(2)$};

\node[draw, thick, shape=circle,scale=1.3] (r) at (6,6) {$r$};
\node[draw, shape=circle,scale=1] (u1) at (4,4) {$u_1$};
\node[draw, shape=circle,scale=1] (u2) at (8,4) {$u_2$};
\node[draw, shape=circle,scale=1] (u3) at (10,2) {$u_3$};
\node[draw, shape=circle,scale=1] (u4) at (12,0) {$u_4$};

\foreach \from/\to in {r/u1,r/u2,u2/u3,u3/u4}
\path[draw, thick] (\from) -- (\to);

\node[left] at (5,5) {$e_1$};
\node[right] at (7,5) {$e_2$};
\node[right] at (9,3) {$e_3$};
\node[right] at (11,1) {$e_4$};

\node at (6,0) {$T(3)$};

\end{tikzpicture}
\end{center}

\begin{example}
\label{exa:P-T(2)}
\end{example}  
\vspace{-2 mm}
\noindent Consider a P-model (with $c=\mathbf{1}$) 
on the rooted tree $T(2)$, where
the prize values are given by $P = \{0,1,2,2,3\}$.

{\em Prize Assignment (A):} 
Consider the case where the prizes have
been simultaneously assigned to the non-root vertices of $T(2)$ by
$p(u_1,u_2,u_3,u_4,u_5) := (0,1,3,2,2)$ in the obvious way.  
We will use a similar shorthand
notation later for the bijection $c$.  In this case we see that for
budgets of $B = 2,3$, we have $\maxp(2,\mathbf{1},p) = 3$ and 
$\maxp(3,\mathbf{1},p) = 5$, respectively.

{\em Prize Assignment (B):}
Consider now the case where the prizes have been simultaneously
assigned to the non-root vertices of $T(2)$ by
$p'(u_1,u_2,u_3,u_4,u_5) := (1,0,3,2,2)$. In this case we see that for
the same budgets of $B = 2,3$ as in (A), we have $\maxp(2,\mathbf{1},p') = 4$
and $\maxp(3,\mathbf{1},p') = 4$, respectively.

From these assignments we see that for budget $B = 2$, the SS in (A)
is better than the one in (B), and for $B = 3$, the SS in (B) is better
than the one in (A).

\vspace{3 mm}

\begin{example}
\label{exa:P-T(3)}
\end{example}
\vspace{-2 mm} 
\noindent Consider a P-model
on the rooted tree $T(3)$, where
the prize values are given by $P = \{0,0,1,1\}$.

{\em Prize Assignment (A):}
Consider the case where the prizes have been simultaneously assigned
to the non-root vertices of $T(3)$ by 
$p(u_1,u_2,u_3,u_4) := (0,0,1,1)$. 
In this case we see that for budgets of $B = 1,3$, we have
$\maxp(1,\mathbf{1},p) = 0$ and $\maxp(3,\mathbf{1},p) = 2$, respectively.

{\em Prize Assignment (B):}
Consider now the case where the prizes have been simultaneously
assigned to the non-root vertices of $T(3)$ by 
$p'(u_1,u_2,u_3,u_4) := (1,0,0,1)$. 
In this case we see that for the same budgets of $B = 1,3$
as in (A), we have $\maxp(1,\mathbf{1},p') = 1$ and 
$\maxp(3,\mathbf{1},p') = 1$, respectively.

From these assignments we see that for budget $B = 1$, the SS in (A)
is better than the one in (B), and for $B = 3$, the SS in (B) is better
than the one in (A). From these examples we have the following
observation.
\begin{observation}
\label{obs:P-T(2,3)}
For general prize values $P$, neither of the P-models 
$M = (T(2),I,P)$ nor $M = (T(3),I,P)$ have optimal SSs.
\end{observation}

We will now consider the dual cases of the C-models.
\begin{example}
\label{exa:C-T(2)}
\end{example}
\vspace{-2 mm} 
\noindent Consider a C-model (with $p=\mathbf{1}$)
on the rooted tree $T(2)$, where
the penetration costs are given by $C = \{0,1,1,2,3\}$.

{\em Cost Assignment (A):}
Consider the case where the penetration costs have been simultaneously
assigned to the edges of $T(2)$ by $c(e_1,e_2,e_3,e_4,e_5) :=
(3,2,0,1,1)$. In this case we see that for budgets of $B = 2,4$, we
have $\maxp(2,c,\mathbf{1}) = 1$ and $\maxp(4,c,\mathbf{1}) = 3$, respectively.

{\em Cost Assignment (B):}
Consider now the case where the penetration costs have been
simultaneously assigned to the edges of $T(2)$ by
$c'(e_1,e_2,e_3,e_4,e_5) := (2,3,0,1,1)$. In this case we see that for
the same budgets of $B = 2,4$ as in (A), we have $\maxp(2,c',\mathbf{1}) = 2$
and $\maxp(4,c',\mathbf{1}) = 2$, respectively.

From these assignments we see that for budget $B = 2$, the SS in (A)
is better than the one in (B), and for $B = 4$, the SS in (B) is better
than the one in (A).

\vspace{3 mm}

\begin{example}
\label{exa:C-T(3)}
\end{example}
\vspace{-2 mm} 
\noindent Consider now a C-model 
on the rooted tree $T(3)$, where
the penetration costs are given by $C = \{0,0,1,1\}$.

{\em Cost Assignment (A):}
Consider the case where the penetration costs have been simultaneously
assigned to the edges of $T(3)$ by $c(e_1,e_2,e_3,e_4) :=
(1,1,0,0)$. In this case we see that for budgets of $B = 0,1$, we have
$\maxp(0,c,\mathbf{1}) = 0$ and $\maxp(1,c,\mathbf{1}) = 3$, respectively.

{\em Cost Assignment (B):}
Consider now the case where the penetration costs have been assigned
to the edges of $T(3)$ by $c'(e_1,e_2,e_3,e_4) := (0,1,1,0)$. In this
case we see that for the same budgets of $B = 0,1$ as in (A), we have
$\maxp(0,c',\mathbf{1}) = 1$ and $\maxp(1,c',\mathbf{1}) = 2$, respectively.

From these assignments we see that for budget $B = 0$, the SS in (A)
is better than the one in (B), and for $B = 1$, the SS in (B) is better
than the one in (A). 

\vspace{3 mm}

From these examples we conclude the following.
\begin{observation}
\label{obs:C-T(2,3)}
For general penetration costs $C$, neither of the C-models 
$M = (T(2),C,I)$ nor $M = (T(3),C,I)$ have optimal SSs.
\end{observation}

{\sc Remarks:} (i) Note that in Examples~\ref{exa:P-T(2)} 
and~\ref{exa:C-T(2)} involving the rooted tree $T(2)$, we 
have that the prize assignments to the non-root vertices and 
cost assignments to the corresponding edges sum 
up to a constant vector for both assignments {\em (A)} and {\em (B)}:
\begin{eqnarray*}
(A) & : & p(u_1,u_2,u_3,u_4,u_5) + c(e_1,e_2,e_3,e_4,e_5)
= (0,1,3,2,2) + (3,2,0,1,1) = (3,3,3,3,3), \\
(B) & : & p'(u_1,u_2,u_3,u_4,u_5) + c'(e_1,e_2,e_3,e_4,e_5) 
= (1,0,3,2,2) + (2,3,0,1,1) = (3,3,3,3,3),
\end{eqnarray*}
and similarly for the rooted tree $T(3)$:
\begin{eqnarray*}
(A) & : & p(u_1,u_2,u_3,u_4) + c(e_1,e_2,e_3,e_4)
= (0,0,1,1) + (1,1,0,0) = (1,1,1,1), \\
(B) & : & p'(u_1,u_2,u_3,u_4) + c'(e_1,e_2,e_3,e_4) 
= (1,0,0,1) + (0,1,1,0) = (1,1,1,1).
\end{eqnarray*}
This duality is not a coincidence and will discussed in more detail
in section~\ref{sec:duality}. (ii) Although special cases of
Theorems~\ref{thm:3-cat-SS},~\ref{thm:4-spid-SS},~\ref{thm:3-cat-CSS}
and~\ref{thm:4-spid-CSS},
it is an easy combinatorial exercise to see that both a 
C- or P-model $M = (T,C,P)$, where $T$ is a proper rooted 
subtree of either
$T(2)$ or $T(3)$ does indeed have an optimal SS, and so $T(2)$ and
$T(3)$ are the smallest rooted trees, in either model, with no optimal
SS. This point will also be discussed and stated explicitly in 
sections~\ref{sec:optimal-P-trees} and~\ref{sec:duality}.

\vspace{3 mm}

Consider now a given rooted tree $T$ and another rooted tree $T^{\dagger}$
containing $T$ as a rooted subtree, so $T\subseteq T^{\dagger}$.  Assume that
the P-model $M = (T,I,P)$ has no optimal SS\@.  Extend $M$ to a P-model
on $T^{\dagger}$ by adding a zero prize for each vertex in 
$V(T^{\dagger})\setminus V(T)$, so $P^{\dagger} = P\cup Z$, 
where $Z$ is the multiset consisting of
$|V(T^{\dagger})| - |V(T)|$ copies of $0$. In this case we have the following.
\begin{observation}
\label{obs:P-subtree}
If $M = (T,I,P)$ is a P-model with no optimal SS, and $T^{\dagger}$ contains $T$
as a rooted subtree, then the P-model 
$M^{\dagger} = (T^{\dagger},I,P^{\dagger})$ has no optimal SS.
\end{observation}
\begin{proof}
(Sketch) For any budget consisting of $B = m$ edges and a
SS $(T, \mathbf{1}, p)$, there is a rooted subtree $\tau$ of $T$ with $m$ edges
such that $\pr(\tau,\mathbf{1},p) = \maxp(m,\mathbf{1},p)$. 
Let $\mathbf{1}$ and $p^{\dagger}$ 
be the obvious extensions of $\mathbf{1}$ and $p$ to $T^{\dagger}$, 
by letting $\mathbf{1}(e) =1$
for all $e\in E(T^{\dagger})$ and $p^{\dagger}(u) = 0$ for any $u\in
V(T^{\dagger})\setminus V(T)$. If $\tau'$ is a rooted subtree of 
$T^{\dagger}$ with
$m$ edges, then $\tau'\cap T$ is a rooted subtree of both $T$ and
$T^{\dagger}$ on $m$ or fewer edges. Since any vertex of 
$V(\tau')\setminus V(T)$ has zero prize, we have
\[
\pr(\tau',\mathbf{1},p^{\dagger}) 
= \pr(\tau'\cap T,\mathbf{1},p^{\dagger}) 
= \pr(\tau'\cap T,\mathbf{1},p) \leq 
\maxp(m,\mathbf{1},p),
\]
with equality for $\tau' = \tau$ since 
$\tau\subseteq T\subseteq T^{\dagger}$.  Hence, 
$\maxp(m,\mathbf{1},p^{\dagger}) = \maxp(m,\mathbf{1},p)$, 
and we conclude that
if $M = (T,I,P)$ has no optimal SS, then neither does 
$M^{\dagger} = (T^{\dagger},I,P^{\dagger})$.
\end{proof}

Dually, assume that we have a C-model $M = (T,C,I)$ that has no optimal
SS, and similarly, let $T^{\dagger}$ be a rooted subtree containing $T$ as a
rooted subtree. Extend $M$ to a C-model on $T^{\dagger}$ by adding penetration
costs of $\infty$\footnote{Where here we can choose $\infty$ to be the
number of edges of $T$ plus one, that is, a large number exceeding any
sensible attack budget.}  for each edge of $T^{\dagger}$ that is not in $T$,
so $C^{\dagger} = C\cup Y$, where $Y$ is the multiset consisting of 
$|E(T^{\dagger})| - |E(T)|$ copies of $\infty$. 
\begin{observation}
\label{obs:C-subtree}
If $M = (T,C,I)$ is a C-model with no optimal SS, and $T^{\dagger}$ contains $T$
as a rooted subtree, then the C-model $M^{\dagger} = (T^{\dagger},C^{\dagger},I)$ 
has no optimal SS.
\end{observation}
\begin{proof}
(Sketch) The proof is similar to the one for
Observation~\ref{obs:P-subtree}.  For any budget $B\in\prats$
and a SS $(T, c, \mathbf{1})$ of $M$, there is a rooted subtree
$\tau$ of $T$ with $m$ edges such that 
$\pr(\tau,c,\mathbf{1}) = \maxp(B,c,\mathbf{1})$.  
Let $c^{\dagger}$ and $\mathbf{1}$ be the obvious 
extensions of $c$
and $\mathbf{1}$ to $T^{\dagger}$, by letting $c^{\dagger}(e) =\infty$
for all $e\in E(T^{\dagger})\setminus E(T)$.
If $\tau'$ is a rooted subtree of $T^{\dagger}$ within the
attacker's budget of $B < \infty$, then every edge of $\tau'$ must
be in $T$, and so $\tau'\subseteq T\subseteq T^{\dagger}$.  Since 
$c^{\dagger}$ agrees with $c$ on the edges of $T$ we have
\[
\pr(\tau',c^{\dagger},\mathbf{1}) = \pr(\tau',c,\mathbf{1}) 
\leq \maxp(B,c,\mathbf{1}),
\]
with equality for $\tau' = \tau$.  Hence, 
$\maxp(B,c^{\dagger},\mathbf{1}) = \maxp(B,c,\mathbf{1})$, and we conclude that 
if $M = (T,C,I)$ has no SS, then
neither does $M^{\dagger} = (T^{\dagger}, C^{\dagger},I)$.
\end{proof}
By Observations~\ref{obs:P-T(2,3)},~\ref{obs:C-T(2,3)},~\ref{obs:P-subtree},
and~\ref{obs:C-subtree} we have the following corollary.
\begin{corollary}
\label{cor:forbidden-subtrees}
If $T$ is a rooted tree such that any P- or C-model 
$M = (T,C,P)$ has an optimal SS, then $T$ contains neither $T(2)$ nor
$T(3)$ as rooted subtrees.
\end{corollary}
Let $T$ be a rooted tree such that any CSM $M = (T,C,P)$ has an
optimal SS\@. Assume further that $T$ is not a path rooted at one of
its two leaves.  If $T$ has at least three non-zero levels (we
consider the root $r$ to be the unique level-$0$ vertex), then $T$
must contain $T(3)$ as a rooted subtree and hence, by
Corollary~\ref{cor:forbidden-subtrees}, there is a CSM $M = (T,C,P)$
with no optimal SS, contradicting our assumption on
$T$\@. Consequently, $T$ has at most two non-zero levels.

If $T$ has at most two non-zero levels, and it has two leaves of
distance four apart (with the root $r$ being midways between them),
then neither parent of the leaves is of degree three or more, because
then $T$ has $T(2)$ as a rooted subtree. And, so again, by
Corollary~\ref{cor:forbidden-subtrees}, there is a CSM $M = (T,C,P)$
with no optimal SS\@.  This observation again contradicts our
assumption on $T$\@. As a result, either (i) $T$ has a diameter of three
and is obtained by attaching an arbitrary number of leaves to the end
vertices of a single edge and then rooting it at one of the
end-vertices of the edge, or (ii) $T$ has diameter of four and each
level-one vertex has degree at most two.

Recall that a {\em caterpillar tree\/} is a tree where each vertex is
within distance one of a central path, and that a {\em spider tree\/}
is a tree with one vertex of degree at least three and all other
vertices of degree at most two.
\begin{definition}
\label{def:caterp-spider}
\begin{itemize}
A {\em rooted path\/} is a path rooted at one of its two leaves.

A {\em rooted star\/} is a star rooted at its unique center vertex.

A {\em 3-caterpillar\/} is a caterpillar tree of diameter three.

A {\em rooted 3-caterpillar\/} is a 3-caterpillar rooted at one
of its two center vertices. 

A {\em 4-spider\/} is a spider tree of diameter four with its unique center
vertex of degree at least three.

A {\em rooted 4-spider\/} is a 4-spider rooted at its unique center vertex. 
\end{itemize}
\end{definition}
By Corollary~\ref{cor:forbidden-subtrees} and the discussion just
before Definition~\ref{def:caterp-spider}, we therefore have the
following main theorem of this section.
\begin{theorem}
\label{thm:SS-trees}
If $T$ is a rooted tree such that any P- or C-model 
$M = (T,C,P)$ has an optimal SS, then $T$ is one of the following types:
(i) a rooted path, 
(ii) a rooted star, 
(iii) a rooted 3-caterpillar, or 
(iv) a rooted 4-spider.
\end{theorem}
It remains to be seen whether or not a rooted 3-caterpillar or a rooted
4-spider $T$ is such that any P- or C-model $M = (T,C,P)$ has an
optimal SS\@.  This item will be the main topic of the next two
sections.

\section{P-models with Optimal Security Systems}
\label{sec:optimal-P-trees} 

In this section we prove that if $T$ is one of the four types of
rooted trees mentioned in Theorem~\ref{thm:SS-trees}, then any P-model 
$M = (T,I,P)$ indeed has an optimal SS. The C-models will be discussed
in section~\ref{sec:duality}. We already have that any P-model 
$M = (T,I,P)$ (in fact, any CSM $M = (T,C,P)$), where $T$ is a 
rooted path or a rooted star, does have an optimal SS, so it
suffices to consider rooted 3-caterpillars and rooted 4-spiders.

Let $T$ be a rooted 3-caterpillar on vertices $\{r,u_1,\ldots,u_n\}$
with edges given by
\begin{equation}
\label{eqn:3-cat-label}
E(T) = \{(r,u_1), \ldots, (r,u_k), (u_1,u_{k+1}),\ldots, (u_1,u_n)\},
\end{equation}
where $2\leq k\leq n-1$. As before, we label the edges by the index
of their heads, so $e_i = (r,u_i)$ for $i\in\{1,\ldots,k\}$
and $e_i = (u_1,u_i)$ for $i\in\{k+1,\ldots,n\}$. Our first 
result is the following.
\begin{theorem}
\label{thm:3-cat-SS}
Let $M = (T,I,P)$ be a P-model where $T$ is a rooted 3-caterpillar and
$P = \{p_1,\ldots,p_n\}$ is a multiset of possible prizes indexed
increasingly $p_1\leq p_2\leq\cdots\leq p_n$. Then the SS $(T,\mathbf{1},p)$,
where $p(u_i) = p_i$ for each $i\in \{1,\ldots,n\}$ is an optimal SS
for $M$.
\end{theorem}
\begin{proof} 
Let $B = m\in\{0,1,\ldots,n\}$ be the attacker's budget, that is
the number of edges an adversary can afford to penetrate. We want
to show that $\maxp(m,\mathbf{1},p)\leq \maxp(m,\mathbf{1},p')$ 
for any prize assignment
$p'$ to the vertices of the rooted 3-caterpillar $T$.

Let $\tau\subseteq T$ be a rooted subtree of $T$ on $m$ edges
with $\pr(\tau,\mathbf{1},p) = \maxp(m,\mathbf{1},p)$. 
There are two cases we need to consider.

{\sc First case:} $e_1 \in E(\tau)$. Since all the leaves are
connected to one of the end-vertices of $e_1 = (r,u_1)$, the remaining
$m-1$ edges of $\tau$ must be incident to the $m-1$ maximum prize
vertices, and so 
$\maxp(m,\mathbf{1},p) = \pr(\tau,\mathbf{1},p) = p_n + p_{n-1} + \cdots +
p_{n-m+2} + p_1$.  If $p'$ is another prize assignment to the vertices
of $T$, then $p'(u_1) = p_c$, where $c\in \{1,\ldots,n\}$. Therefore,
$\maxp(m,\mathbf{1},p') \geq \pr(\tau',\mathbf{1},p')$, where 
$\tau'$ is a rooted subtree of
$T$ that contains $e_1$ and contains all the remaining $m-1$ maximum
prizes, and so
\[
\pr(\tau',\mathbf{1},p') = \left\{ 
\begin{array}{ll}
p_n + p_{n-1} + \cdots + p_{n-m+1} & \mbox{ if } c\in \{n-m+1,\ldots,n\}, \\
p_n + p_{n-1} + \cdots + p_{n-m+2} + p_c
  & \mbox{ if } c\not\in \{n-m+1,\ldots,n\}. 
\end{array}
\right.
\]
In either case we have 
$\pr(\tau',\mathbf{1},p') \geq p_n + p_{n-1} + 
\cdots p_{n-m+2} + p_1 = \maxp(m,\mathbf{1},p)$, 
and so $\maxp(m,\mathbf{1},p') \geq \maxp(m,\mathbf{1},p)$ in this case.

{\sc Second case:} $e_1\not\in E(\tau)$. For this case to be
possible we must have $m\leq k-1$, since otherwise $e_1$ must be
in $\tau$. Secondly, we must have that $\tau$ contains all the maximum
prize vertices on level one and so 
$\maxp(m,\mathbf{1},p) = \pr(\tau,\mathbf{1},p) 
= p_k + p_{k-1} + \cdots + p_{k-m+1}$.
In particular, we must have 
\[
p_k + p_{k-1} + \cdots + p_{k-m+1} \geq p_n + p_{n-1} + \cdots + p_{n-m+2} + p_1,
\]
since a tree containing $e_1$ does not have a greater total prize than
$\tau$.  If $p'$ is another prize assignment to the vertices of $T$,
then let $\{{\ell}_1,\ldots,{\ell}_k\}$ be the indices of the prizes
assigned to vertices on level one by $p'$, that is,
$\{p_{{\ell}_1},\ldots,p_{{\ell}_k}\} = \{p'(u_1),\ldots, p'(u_k)\}$
as multisets. If now $\tau'$ is the rooted subtree of $T$ with $m$ edges
containing the $m$ vertices with the largest prizes, then, since
$p_{{\ell}_i}\geq p_i$ for each $i\in\{1,\ldots,k\}$, we have
\[
\maxp(m,\mathbf{1},p')\geq \pr(\tau',\mathbf{1},p') 
= p_{{\ell}_k} + p_{{\ell}_{k-1}} + \cdots + p_{{\ell}_{k-m+1}} 
\geq p_k + p_{k-1} + \cdots + p_{k-m+1} = \maxp(m,\mathbf{1},p),
\]
in this case as well. This completes the proof that
the SS $(T,p)$ is optimal.
\end{proof}

Now, let $T$ be a rooted 4-spider on vertices $\{r,u_1,\ldots,u_n\}$
with edges given by
\begin{equation}
\label{eqn:4-spider-label}
E(T) = \{(r,u_1),\ldots, (r,u_k), 
(u_1,u_{k+1}),(u_2,u_{k+2}),\ldots, (u_{n-k} ,u_n)\},
\end{equation}
where $n/2 \leq k\leq n-2$. As before, the edges are labeled by the
index of their heads: $e_i = (r,u_i)$ for $i\in \{1,\ldots,k\}$ and
$e_i = (u_{i-k},u_i)$ for $i\in\{k+1,\ldots,n\}$.  Our second result
is the following.
\begin{theorem}
\label{thm:4-spid-SS}
Let $M = (T,I,P)$ be a P-model, where $T$ is a rooted 4-spider and 
$P = \{p_1,\ldots,p_n\}$ is a multiset of possible prizes indexed
increasingly $p_1\leq p_2\leq\cdots\leq p_n$. Then the SS $(T,\mathbf{1},p)$,
where $p(u_i) = p_i$ for $i\in \{1,\ldots,k\}$ and 
$p(u_i) = p_{n+k+1-i}$ for $i\in\{k+1,\ldots,n\}$ is an optimal SS for $M$.
\end{theorem}
Before we prove Theorem~\ref{thm:4-spid-SS}, we need a few
lemmas that will come in handy for the proof.
\begin{lemma}
\label{lmm:swapi1jleaf}
Let $T$ be a 4-spider presented as in (\ref{eqn:4-spider-label}) and
$m\in\nats$.  Let $p$ be a prize assignment on $V(T)$ such that $p_i =
p(u_i) \leq p(u_j) = p_j$, where $u_i$ is on level one and $u_j$ is a
leaf of $T$\@.  If $p'$ is the prize assignment obtained from $p$ by
swapping the prizes of $u_i$ and $u_j$, then 
$\maxp(m,\mathbf{1},p)\leq \maxp(m,\mathbf{1},p')$.
\end{lemma}
\begin{proof}
If $j=k+i$, so $u_j$ is the unique child of $u_i$, then 
the lemma holds by (\ref{eqn:p'leqp}).
Hence, we can assume that $u_j$ is not a child of $u_i$.
Let $\tau\subseteq T$ be a max-prize rooted subtree on $m$
edges, so $\pr(\tau,\mathbf{1},p) = \maxp(m,\mathbf{1},p)$. 
We now consider the following cases.

If either both $u_i$ and $u_j$ are vertices of $\tau$, or neither of
them are, then clearly $\maxp(m,\mathbf{1},p) = 
\pr(\tau,\mathbf{1},p) = \pr(\tau,\mathbf{1},p')\leq \maxp(m,\mathbf{1},p')$.

If $u_i\in V(\tau)$ and $u_j\not\in V(\tau)$,
then 
\[
\maxp(m,\mathbf{1},p) = \pr(\tau,\mathbf{1},p) \leq 
\pr(\tau,\mathbf{1},p)-p_i+p_j = \pr(\tau,\mathbf{1},p')
\leq \maxp(m,\mathbf{1},p).
\]

If $u_i\not\in V(\tau)$ and $u_j\in V(\tau)$, then, since $u_i$ is on
level one and $u_j$ is a leaf of $\tau$, we have that $\tau' = (\tau -
u_j)\cup u_i$ is also a rooted subtree of $T$ on $m$ vertices and
$\maxp(m,\mathbf{1},p) = \pr(\tau,\mathbf{1},p) 
= \pr(\tau',\mathbf{1},p')\leq \maxp(m,\mathbf{1},p')$, which
completes our proof.
\end{proof}
Let $M = (T,I,P)$ be a P-model where $T$ is a rooted 4-spider, 
$P = \{p_1,\ldots,p_n\}$, and $p'$ be an arbitrary prize assignment on
$V(T)$.  Since every vertex of $T$ on level two is automatically a
leaf, we can, by repeated use of Lemma~\ref{lmm:swapi1jleaf}, obtain a
prize assignment with smaller max-prize with respect to any $m$ that
has its $n-k$ largest prizes on its level-two vertices, and hence has
its $k$ smallest prizes on the level-one vertices $u_1,\ldots, u_k$ of
$T$\@.  By further use of the same Lemma~\ref{lmm:swapi1jleaf} when
considering these level-one vertices of $T$, we can obtain a prize
assignment $p$ that has its smallest prizes on the non-leaf vertices
on level one and yet with smaller max-prize, so 
$\maxp(m,\mathbf{1},p)\leq \maxp(m,\mathbf{1},p')$ for any $m$. 
Note that our $p$ satisfies
\[
p(\{u_1,\ldots,u_{n-k}\}) = \{p_1,\ldots,p_{n-k}\}, \ \ 
p(\{u_{k+1},\ldots,u_n\}) = \{p_{k+1},\ldots,p_n\}.
\]
As the level-one vertices of $T$ can be assumed to be ordered by 
their prizes, we summarize in the following.
\begin{corollary}
\label{cor:max-parts}
From any prize assignment $p'$ we can by repeated use 
of Lemma~\ref{lmm:swapi1jleaf} obtain a prize assignment
$p$ on our 4-spider $T$, presented as in
(\ref{eqn:4-spider-label}), such that 
\[
p(u_i) = p_i \mbox{ for all }i\in\{1,\ldots,k\}, \mbox{ and }
p(u_i) = p_{\pi(i)} \mbox{ for all } i\in\{k+1,\ldots,n\},
\]
where $\pi$ is a permutation of $\{k+1,\ldots,n\}$,
and with $\maxp(m,\mathbf{1},p)\leq \maxp(m,\mathbf{1},p')$ 
for any $m\in\nats$.
\end{corollary}
Our next lemma provides our final tool in proving
Theorem~\ref{thm:4-spid-SS}.
\begin{lemma}
\label{lmm:xy-ordered}
Let $T$ be a 4-spider presented as in (\ref{eqn:4-spider-label}) and
$m\in\nats$.  Let $p$ be a prize assignment on $V(T)$ such that for
some $i,j\in \{1,\ldots,n-k\}$ with $i<j$, we have $p(u_i) \leq
p(u_j)$ and $p(u_{i+k}) \geq p(u_{j+k})$.  If $p'$ is a prize
assignment where the prizes on $u_{i+k}$ and $u_{j+k}$ have been
swapped, then $\maxp(m,\mathbf{1},p)\leq \maxp(m,\mathbf{1},p')$.
\end{lemma}
\begin{proof}
Let $\tau\subseteq T$ be a max-prize rooted subtree on $m$
edges with respect to $p$, so 
$\pr(\tau,\mathbf{1},p) = \maxp(m,\mathbf{1},p)$. We now 
consider the following cases.

If either both $u_{i+k}$ and $u_{j+k}$ are vertices of $\tau$, or
neither of them are, then clearly 
$\maxp(m,\mathbf{1},p) = \pr(\tau,\mathbf{1},p) 
= \pr(\tau,\mathbf{1},p')\leq \maxp(m,\mathbf{1},p')$.

If $u_{i+k}\not\in V(\tau)$ and $u_{j+k}\in V(\tau)$, then 
\[
\maxp(m,\mathbf{1},p) = \pr(\tau,\mathbf{1},p) 
\leq \pr(\tau,\mathbf{1},p)-p(u_{j+k})+p(u_{i+k}) 
= \pr(\tau,\mathbf{1},p')\leq \maxp(m,\mathbf{1},p').
\]

If $u_{i+k}\in V(\tau)$ and $u_{j+k}\not\in V(\tau)$, then we
consider two (sub-)cases.
If $u_j\in V(\tau)$, then since $u_j$ is a leaf in $\tau$,
we have that $\tau' = (\tau - u_{i+k})\cup u_{j+k}$ is
also a rooted subtree of $T$ on $m$ vertices and
$\maxp(m,\mathbf{1},p) = \pr(\tau,\mathbf{1},p)
= \pr(\tau',\mathbf{1},p')\leq \maxp(m,\mathbf{1},p')$.
If $u_j\not\in V(\tau)$, then 
$\tau'' = (\tau - \{u_i,u_{i+k}\})\cup \{u_j,u_{j+k}\}$
is also a rooted subtree of $T$ on $m$ vertices, and
\begin{eqnarray*}
\maxp(m,\mathbf{1},p) & = & \pr(\tau,\mathbf{1},p) \\
  & \leq & \pr(\tau,\mathbf{1},p) - p(u_i) - p(u_{j+k}) + p(u_j) + p(u_{i+k}) \\ 
  & =    & \pr(\tau'',\mathbf{1},p') \\ 
  & \leq & \maxp(m,\mathbf{1},p'),
\end{eqnarray*}
which completes the proof.
\end{proof}
\begin{proof}[Proof of Theorem~\ref{thm:4-spid-SS}]
Let $T$ be a 4-spider, $p$ a prize assignment as given in
Theorem~\ref{thm:4-spid-SS}, and $m\in\nats$.  Let $p'$ be an
arbitrary prize assignment of the vertices of $T$\@. By
Corollary~\ref{cor:max-parts} we can obtain a prize assignment $p''$
such that
\[
p''(u_i) = p_i \mbox{ for all }i\in\{1,\ldots,k\}, \mbox{ and }
p''(u_i) = p_{\pi(i)} \mbox{ for all } i \in \{k+1,\ldots,n\},
\]
where $\pi$ is a permutation of $\{k+1,\ldots,n\}$, and with
$\maxp(m,\mathbf{1},p'')\leq \maxp(m,\mathbf{1},p')$ for any $m\in\nats$.  By
Lemma~\ref{lmm:xy-ordered} we can obtain a prize assignment $p$ on
$V(T)$ from $p''$ simply by ordering the prizes on the level-two
leaves in a decreasing order, thereby obtaining the very prize
assignment $p$ from Theorem~\ref{thm:4-spid-SS} that satisfies
$\maxp(m,\mathbf{1},p)\leq \maxp(m,\mathbf{1},p'')$ for any $m\in\nats$.  
This proves that
for any $m\in\nats$ we have 
$\maxp(m,\mathbf{1},p)\leq \maxp(m,\mathbf{1},p'')\leq \maxp(m,\mathbf{1},p')$, 
and since $p'$ was an arbitrary prize assignment, the proof is complete.
\end{proof}
As a further observation, we can describe the optimal SAs on the
P-model $M = (T,I,P)$, where $T$ is a rooted 4-spider with the vertices
and edges labeled as in (\ref{eqn:4-spider-label}), as follows.
\begin{observation}
\label{obs:opt-4-spider}
Let $T$ be a 4-spider, $p$ a prize assignment as in 
Theorem~\ref{thm:4-spid-SS}, and $m\in\nats$.
Then there is a max-prize rooted subtree 
$\tau\subseteq T$ on $m$ edges with respect to $p$, so
$\pr(\tau,\mathbf{1},p) = \maxp(m,\mathbf{1},p)$, 
with the following property:
\begin{enumerate}
  \item If $n\leq 2k-1$, then all the leaves of $\tau$ 
are leaves in $T$, and hence in $\{u_{n-k+1},\ldots,u_n\}$.
  \item If $n = 2k$, then $\tau$ has at most one leaf on
level one, in which case it can assumed to be $u_k$.
\end{enumerate}
\end{observation}
\begin{proof}
Suppose $\tau$ has two leaves $u_i,u_j\in\{u_1,\ldots,u_{n-k}\}$.  In
this case $\tau' = (\tau - u_j)\cup u_{k+i}$ is also a rooted subtree
of $T$ on $m$ edges and has 
$\pr(\tau',\mathbf{1},p)\geq \pr(\tau,\mathbf{1},p)$. Hence, we
can assume $\tau$ to have at most one leaf from
$\{u_1,\ldots,u_{n-k}\}$.

Suppose $\tau$ has one leaf $u_i\in\{u_1,\ldots,u_{n-k}\}$.  We now
consider the two cases; $k>n-k$ and $k=n-k$.

{\sc First case:} $k>n-k$ or $n\leq 2k-1$. If $\tau$ has another leaf
$u_j\in\{u_{n-k+1},\ldots,n_k\}$, then, as above, 
$\tau' = (\tau - u_j)\cup u_{k+i}$ has 
$\pr(\tau',\mathbf{1},p)\geq \pr(\tau,\mathbf{1},p)$.
Otherwise, $\tau$ has no leaves from $\{u_{n-k+1},\ldots,n_k\}\neq\emptyset$.
In this case $\tau'' = (\tau - u_i)\cup u_k$ is a rooted subtree of $T$
on $m$ edges with $\pr(\tau'',\mathbf{1},p)\geq \pr(\tau,\mathbf{1},p)$.
Hence, we can assume that $\tau$ has no leaves 
from $\{u_1,\ldots,u_{n-k}\}$, which proves or claim in this case.

{\sc Second case:} $k=n-k$ or $n=2k$. In this case $\tau$ has the
unique level-one leaf $u_i$. If $i<k$, then $u_k$ has a unique child
$u_{2k}$ in $\tau$, and so $\tau' = (\tau - u_{2k})\cup u_{k+i}$ has
the unique level-one leaf $u_k$ and 
$\pr(\tau',\mathbf{1},p)\geq \pr(\tau,\mathbf{1},p)$.
Hence, we can assume that $\tau$ has its unique level-one leaf $u_k$.
\end{proof}

{\sc Remark:} Note that in the case $n\leq 2k-1$ in the proof of
Observation~\ref{obs:opt-4-spider}, all the level-one leaves of $\tau$
can be assumed to be from $\{u_{n-k+1},\ldots, u_k\}$.  If we have $\ell$
of them, then they can further be assumed to be $u_{k-\ell+1},\ldots,u_k$.

\section{Duality between P- and C-Models}
\label{sec:duality} 

In this section we state and use a duality between
the P- and C-models, which then can be used to obtain similar 
results for C-models that we obtained for P-models in the previous section. 
In particular, we will demonstrate that if $T$ is one of the four types of
rooted trees mentioned in Theorem~\ref{thm:SS-trees}, then any 
C-model $M = (T,C,{{I}})$ indeed has an optimal SS, as we proved was
the case for the P-model. As with the P-model, we already have that any
C-model $M = (T,C,I)$ (in fact, any CSM $M = (T,C,P)$), where $T$
is a rooted path or a rooted star, does have an optimal SS.

As mentioned in the remarks right after Observation~\ref{obs:C-T(2,3)},
we now explicitly examine an example of a rooted proper subtree 
$T_p(2)$ of $T(2)$, for which any P- or C-model $M = (T_p(2),C,P)$ 
has an optimal security system. 
For the next two examples, and just as in the convention right before
Example~\ref{exa:P-T_p(3)}, let $T_p(2)$ denote the 
rooted tree, whose underlying graph is a path, on five vertices 
$V(T_p(2)) = \{r,u_1,u_2,u_3,u_4\}$
and edges $E(T_p(2)) = \{ (r,u_1), (r,u_2), (u_1,u_3), (u_2,u_4)\}$
rooted at its center vertex.
We continue the convention of labeling the edges by the same index
as their heads: $e_1 = (r,u_1)$, $e_2 = (r,u_2)$, $e_3 = (u_1,u_3)$,
and $e_4 = (u_2,u_4)$.
\begin{center}
\begin{tikzpicture}[scale=.7] 

\node[draw, thick, shape=circle,scale=1.3] (r) at (0,6) {$r$};
\node[draw, shape=circle,scale=1] (u1) at (-2,4) {$u_1$};
\node[draw, shape=circle,scale=1] (u2) at (2,4) {$u_2$};
\node[draw, shape=circle,scale=1] (u3) at (-4,2) {$u_3$};
\node[draw, shape=circle,scale=1] (u4) at (4,2) {$u_4$};

\foreach \from/\to in {r/u1,r/u2,u1/u3,u2/u4}
\path[draw, thick] (\from) -- (\to);

\node[left] at (-1,5) {$e_1$};
\node[right] at (1,5) {$e_2$};
\node[left] at (-3,3) {$e_3$};
\node[right] at (3,3) {$e_4$};

\node at (0,0) {$T_p(2)$};

\end{tikzpicture}
\end{center}
\begin{example}
\label{exa:P-T_p(2)}
\end{example}
\vspace{-2 mm} 
\noindent Consider a P-model (with $c=\mathbf{1}$)
on the rooted tree $T_p(2)$ where the prize values 
$P = \{p_1,p_2,p_3,p_4\}$ are general real positive values
ordered increasingly $p_1\leq p_2\leq p_3\leq p_4$. By 
Theorem~\ref{thm:4-spid-SS} an optimal SS for our CSM $M=(T_p(2),{{I}},P)$
is obtained by assigning the prizes as
$p(u_1,u_2,u_3,u_4) := (p_1,p_2,p_4,p_3)$. We can explicitly 
obtain the max-prize subtree for each given budgets $B\in {\reals}$
that yields the following:
\[
\maxp(B,\mathbf{1},p) = \left\{
\begin{array}{ll}
0   & \mbox{ for } B < 1, \\
p_2 & \mbox{ for } 1\leq B < 2, \\
\max(p_1+p_4, p_2+p_3) & \mbox{ for } 2\leq B < 3, \\
p_1+p_2+p_4 & \mbox{ for } 3\leq B < 4, \\
p_1+p_2+p_3+p_4 & \mbox{ for } 4\leq B.
\end{array}
\right.
\]
\begin{example}
\label{exa:C-T_p(2)}
\end{example}
\vspace{-2 mm} 
\noindent Consider a C-model (with $p=\mathbf{1}$)
on the rooted tree $T_p(2)$ where the penetration cost values 
$C = \{c_1,c_2,c_3,c_4\}$ are general real positive values
ordered decreasingly $c_1\geq c_2\geq c_3\geq c_4$. It is now
an easy combinatorial exercise to verify directly that 
an optimal SS for our CSM $M=(T_p(2),C,{{I}})$
can be obtained by assigning penetration costs as
$c(u_1,u_2,u_3,u_4) := (c_1,c_2,c_4,c_3)$, in the same (index-)order
as for the P-model in Example~\ref{exa:P-T_p(2)}. We explicitly 
obtain the max-prize subtree for each given budget $B\in {\reals}$
that yields the following:
\[
\maxp(B,c,\mathbf{1}) = \left\{
\begin{array}{ll}
0   & \mbox{ for } B < c_2, \\
1   & \mbox{ for } c_2 \leq B < \min(c_1+c_4, c_2+c_3), \\
2   & \mbox{ for } \min(c_1+c_4, c_2+c_3) \leq B < c_1+c_2+c_4, \\
3   & \mbox{ for } c_1+c_2+c_4\leq B < c_1+c_2+c_3+c_4, \\
4   & \mbox{ for } c_1+c_2+c_3+c_4\leq B.
\end{array}
\right.
\]
Let $K$ be a sufficiently large cost number (any real number 
$\geq \max(c_1,\ldots,c_4) + 1$ will do), 
and write each edge-cost of the form  $c_i = K - c_i'$.
In this way $\maxp(B,c,\mathbf{1})$ will take the following form
\[
\maxp(B,c,\mathbf{1}) = \left\{
\begin{array}{ll}
0   & \mbox{ for } B < K - c_2', \\
1   & \mbox{ for } K - c_2' \leq B < 2K - \max(c_1'+c_4', c_2'+c_3'), \\
2   & \mbox{ for } 2K - \max(c_1'+c_4', c_2'+c_3') \leq B < 
   3K - (c_1'+c_2'+c_4'), \\
3   & \mbox{ for } 3K - (c_1'+c_2'+c_4') \leq B < 4K - (c_1'+c_2'+c_3'+c_4'), \\
4   & \mbox{ for } 4K - (c_1'+c_2'+c_3'+c_4') \leq B.
\end{array}
\right.
\]
From the above we see the evident resemblance to the expression for 
$\maxp(B,\mathbf{1},p)$ 
of the P-model in Example~\ref{exa:P-T_p(2)}. This is a glimpse of a duality
between the P-models and the C-models that we will now describe.

{\sc Convention:} In what follows, it will be convenient to view
the cost and prize assignments $c$ and $p$ not as functions as in 
Definition~\ref{def:GO-defensive}, but rather as vectors
$\tilde{c} = (c_1,\ldots,c_n)$ and $\tilde{p} = (p_1,\ldots,p_n)$ 
in the $n$-dimensional Euclidean space ${\reals}^n$, which can
be obtained by a fixed labeling of the $n$ non-root vertices 
$u_1,\ldots,u_n$ and a corresponding labeling of the edges
$e_1,\ldots,e_n$, with our usual convention that for each $i$ the vertex
$u_i$ is the head of $e_i$, and by letting $c_i := c(e_i)$ and $p_i := p(u_i)$.

\vspace{3 mm}

For a given $n\in{\nats}$, let ${\cal{B}}({\reals}^n)$ denote
the group of all bijections ${\reals}^n \rightarrow {\reals}^n$ with
respect to compositions of maps. For $a\in{\prats}$ and $b\in{\rats}$
the affine map $\alpha : {\reals}^n \rightarrow {\reals}^n$ given
by $\alpha(\tilde{x}) = a\tilde{x} + b\tilde{1}$, where 
$\tilde{1} = (1,\ldots,1)\in {\reals}^n$, is bijective with
an inverse 
$\alpha^{-1}(\tilde{x}) = \frac{1}{a}\tilde{x} - \frac{b}{a}\tilde{1}$
of the same type.
Further, if $\alpha'(\tilde{x}) = a'\tilde{x} + b'\tilde{1}$ is 
another such map,
then the composition 
$(\alpha'\circ\alpha)(\tilde{x}) = a'a\tilde{x} + (a'b + b')\tilde{1}$
is also a bijection of this very type. Since the identity map of ${\reals}^n$ 
has $a = 1\in {\prats}$ and $b = 0 \in \rats$, we have the following.
\begin{observation}
\label{obs:subgroup}
If $n\in {\nats}$ then $G_n = \{\alpha \in {\cal{B}}({\reals}^n) : 
\alpha(\tilde{x}) = a\tilde{x} + b\tilde{1}, \mbox{ for some } a\in{\prats}
\mbox{ and } b\in{\rats}\}$ is a subgroup of ${\cal{B}}({\reals}^n)$.
\end{observation}
By letting $G_n$ act on the set ${\reals}^n$ in the natural way,
$(\alpha, \tilde{x}) \mapsto \alpha(\tilde{x})$, then the group orbits
$G_n(\tilde{x}) = \{ \alpha(\tilde{x}) : \alpha \in G_n\}$ yield a partition 
of ${\reals}^n$ into corresponding equivalence classes 
${\reals}^n = \bigcup_{\tilde{x}\in {\reals}^n}G_n(\tilde{x})$.
By intersecting with ${\prats^n}$ we obtain
the following equivalence classes that we seek.
\begin{definition}
\label{def:equiv-class}
For each $\tilde{x}\in {\prats^n}$ let $[\tilde{x}]$ denote
the equivalence class of $\tilde{x}$ with respect to the 
partition of ${{\reals}}^n$ into the $G_n$ orbits: 
$[\tilde{x}] = G_n(\tilde{x})\cap {\prats^n}$.
\end{definition}
We now justify the above equivalence of vectors of ${\prats^n}$.
The following observation is obtained directly from  
Definition~\ref{def:GO-defensive}.
\begin{observation}
\label{obs:alpha-equiv}
Let $T$ be a rooted tree on $n$ labeled non-root vertices and edges, $\tau$
a rooted subtree of $T$, and $\alpha\in G_n$ given by 
$\alpha(\tilde{x}) = a\tilde{x} + b\tilde{1}$. 
If $\tilde{c}, \tilde{p}\in {\prats^n}$ are a cost and prize
vector, respectively, then we have
\begin{eqnarray*}
\pr(\tau,\tilde{c},\alpha(\tilde{p})) 
  & =  & a\pr(\tau,\tilde{c},\tilde{p}) + |E(\tau)|b, \\
\cost(\tau,\alpha(\tilde{c}),\tilde{p}) 
  & =  & a\cost(\tau,\tilde{c},\tilde{p}) + |E(\tau)|b.
\end{eqnarray*}
\end{observation}
If $J\subseteq \{1,\ldots,n\}$ and $\Sigma_J : {\reals}^n \rightarrow \reals$
is given by $\tilde{x}\mapsto \sum_{i\in J}x_i$, then we clearly have
\begin{equation}
\label{eqn:order-preserving}
\Sigma_J(\alpha(\tilde{x}))\leq \Sigma_J(\alpha(\tilde{y})) 
\Leftrightarrow \Sigma_J(\tilde{x})\leq \Sigma_J(\tilde{y}),
\end{equation}
and hence the following corollary.
\begin{corollary}
\label{cor:alpha-max}
Let $T$ be a rooted tree on $n$ labeled non-root vertices and edges, 
$B\in {\prats}$ a budget, and $\alpha\in G_n$ given by 
$\alpha(\tilde{x}) = a\tilde{x} + b\tilde{1}$. 

(i) If $\tilde{p}\in {\prats^n}$ is a prize vector, then we have
\begin{equation}
\label{eqn:alpha-P}
\maxp(B,\tilde{1},\alpha(\tilde{p})) = 
a\maxp(B,\tilde{1},\tilde{p}) + b\lfloor B\rfloor.
\end{equation}
Further, both max prizes in (\ref{eqn:alpha-P}) are attained at the 
same rooted subtree $\tau$ of $T$ where $|E(\tau)| = \lfloor B\rfloor$. 

(ii) If $\tilde{c} \in {\prats^n}$ is a cost vector, then we have
\[
\maxp(aB+bm,\alpha(\tilde{c}), \tilde{1}) = m
\Leftrightarrow
\maxp(B,\tilde{c}, \tilde{1}) = m,
\]
and further, both max prizes are attained at the same rooted subtree $\tau$
of $T$ within the budget; that is, $|E(\tau)| = m$ and 
$\cost(\tau,\tilde{c},\tilde{1}) \leq B$.
\end{corollary} 

{\sc Remarks:} (i) That both max prizes are attained at the same rooted
subtree $\tau$ in (i) in Corollary~\ref{cor:alpha-max} simply
means that 
\[
\pr(\tau,\tilde{1},\alpha(\tilde{p})) = 
\maxp(B,\tilde{1},\alpha(\tilde{p})) 
\Leftrightarrow 
\pr(\tau,\tilde{1},\tilde{p}) = \maxp(B,\tilde{1},\tilde{p}),
\]
which is a direct consequence of Observation~\ref{obs:alpha-equiv}
and (\ref{eqn:alpha-P}).
(ii) Also, for a rooted subtree $\tau$ with $|E(\tau)| = m$ and 
$\cost(\tau,\tilde{c},\tilde{1}) \leq B$, then by
Observation~\ref{obs:alpha-equiv} we also have
$\cost(\tau,\alpha(\tilde{c}),\tilde{1}) \leq aB + bm$, and 
\[
\pr(\tau,\tilde{c},\tilde{1}) = m = \maxp(B,\tilde{c},\tilde{1})
\Leftrightarrow
\pr(\tau,\alpha(\tilde{c}),\tilde{1}) = m 
= \maxp(aB+bm,\alpha(\tilde{c}),\tilde{1}).
\]

We can, in fact, say a tad more than Corollary~\ref{cor:alpha-max}
for C-models $M=(T,C,{{I}})$.
\begin{definition}
\label{def:B_m(c)}
Let $M=(T,C,{{I}})$ be a C-model. 
For a given cost vector $\tilde{c}\in{\prats^n}$ let $B_m(\tilde{c})$
denote the smallest cost $B\in\prats$ with 
$\maxp(B,\tilde{c},\tilde{1}) = m$.
\end{definition}
Note that
\[
\maxp(B,\tilde{c},\tilde{1}) = m \Leftrightarrow
B_m(\tilde{c}) \leq B < B_{m+1}(\tilde{c}).
\]
We also have the following useful lemma.
\begin{lemma}
\label{lmm:B_m(alpha(c))}
If $\alpha\in G_n$ is given by $\alpha(\tilde{x}) = a\tilde{x} + b\tilde{1}$,
then $B_m(\alpha(\tilde{c})) = aB_m(\tilde{c}) + bm$.
\end{lemma}
\begin{proof}
By definition of $B_m(\tilde{c})$ we have
$\maxp(B_m(\tilde{c}),\tilde{c},\tilde{1}) = m$, and hence
by Corollary~\ref{cor:alpha-max} 
$\maxp(aB_m(\tilde{c})+bm,\alpha(\tilde{c}),\tilde{1}) = m$ as well.
Suppose that $\maxp(B',\alpha(\tilde{c}),\tilde{1}) = m$, where
$B'< aB_m(\tilde{c})+bm$. If now $B' = aB'' + bm$, then $B'' < B_m(\tilde{c})$
and we have again by Corollary~\ref{cor:alpha-max} that 
$\maxp(B'',\tilde{c},\tilde{1}) = m$. This contradicts the
definition of $B_m(\tilde{c})$. Hence, 
$B_m(\alpha(\tilde{c}) = aB_m(\tilde{c}) + bm$.
\end{proof}
\begin{proposition}
\label{prp:B_m(c)-equiv}
For $m\in\{0,1,\ldots,n\}$ and a cost vectors $\tilde{c}$ and $\tilde{c}'$
we have $B_m(\tilde{c})\geq B_m(\tilde{c}')$ if and only if for
every budget $B$ with $\maxp(B,\tilde{c},\tilde{1}) = m$ we
have $\maxp(B,\tilde{c},\tilde{1})\leq\maxp(B,\tilde{c}',\tilde{1})$.
\end{proposition}
\begin{proof}
Suppose $B_m(\tilde{c})\geq B_m(\tilde{c}')$, and let $B$ be a budget
with $\maxp(B,\tilde{c},\tilde{1}) = m$. By definition we then 
have $B\geq B_m(\tilde{c})$ and hence $B\geq B_m(\tilde{c}')$
and therefore 
$\maxp(B,\tilde{c}',\tilde{1})\geq m = \maxp(B,\tilde{c},\tilde{1})$.

Conversely, if for every budget $B$ with $\maxp(B,\tilde{c},\tilde{1}) = m$ we
have $\maxp(B,\tilde{c},\tilde{1})\leq\maxp(B,\tilde{c}',\tilde{1})$,
then, in particular for $B = B_m(\tilde{c})$ we have 
$m = \maxp(B_m(\tilde{c}),\tilde{c},\tilde{1})\leq 
\maxp(B_m(\tilde{c}),\tilde{c}',\tilde{1})$, and hence, by definition,
$B_m(\tilde{c}')\leq B_m(\tilde{c})$.
\end{proof}
{\sc Convention:} For a vector $\tilde{x} = (x_1,\ldots, x_n) \in {\prats^n}$ 
let $\{\tilde{x}\}$ denote its underlying multiset. So if 
$(T,\tilde{c},\tilde{p})$
is an SS for a CSM $M = (T,C,P)$, then we necessarily have
$C = \{\tilde{c}\}$ and $P = \{\tilde{p}\}$ as multisets. Also,
we have $\{\tilde{1}\} = {{I}}$ as the multiset containing $n$
copies of $1$.

Suppose $\maxp(B,\tilde{1},\tilde{p}) \leq \maxp(B,\tilde{1},\tilde{p}')$
for all $\tilde{p}'$ with $\{\tilde{p}'\} = \{\tilde{p}\}$. Then
by Corollary~\ref{cor:alpha-max} we get for any $\alpha\in G_n$
with $\alpha(\tilde{x}) = a\tilde{x} + b\tilde{1}$, that 
\[
\maxp(B,\tilde{1},\alpha(\tilde{p})) = 
a\maxp(B,\tilde{1},\tilde{p})  + b\lfloor B\rfloor \leq 
a\maxp(B,\tilde{1},\tilde{p}') + b\lfloor B\rfloor =
\maxp(B,\tilde{1},\alpha(\tilde{p}')),
\]
and so we have the following.
\begin{proposition}
\label{prp:P-optimal}
The SS $(T,\tilde{1},\tilde{p})$ is optimal for the 
P-model $M = (T,{{I}},\{\tilde{p}\})$ with respect to the 
budget $B\in\prats$ if and only if 
the SS $(T,\tilde{1},\alpha(\tilde{p}))$ is optimal for the
P-model $M = (T,{{I}},\{\alpha(\tilde{p})\})$ with respect to $B$.
\end{proposition}

In a similar way, we have by Proposition~\ref{prp:B_m(c)-equiv} that
$\maxp(B,\tilde{c},\tilde{1}) = m \leq \maxp(B,\tilde{c}',\tilde{1})$
whenever $B_m(\tilde{c})\leq B < B_{m+1}(\tilde{c})$ and 
$\{\tilde{c}'\} = \{\tilde{c}\}$ if and only if 
$B_m(\tilde{c})\geq B_m(\tilde{c}')$, which by
Lemma~\ref{lmm:B_m(alpha(c))} holds if and only if
\[
B_m(\alpha(\tilde{c})) = aB_m(\tilde{c}) + bm 
  \geq aB_m(\tilde{c}') + bm = B_m(\alpha(\tilde{c}')).
\]
In other words, 
$\maxp(B,\tilde{c},\tilde{1})\leq \maxp(B,\tilde{c}',\tilde{1})$
when $B_m(\tilde{c})\leq B < B_{m+1}(\tilde{c})$ holds if and only if
$\maxp(B',\alpha(\tilde{c}),\tilde{1})\leq 
\maxp(B',\alpha(\tilde{c}'),\tilde{1})$
when $B_m(\alpha(\tilde{c}))\leq B' < B_{m+1}(\alpha(\tilde{c}))$.
Since this holds for every $\alpha\in G_n$, which is a group with
each element having an inverse, then we have the following.
\begin{proposition}
\label{prp:C-optimal}
The SS $(T,\tilde{c},\tilde{1})$ is optimal for the 
C-model $M = (T,\{\tilde{c}\},{{I}})$ 
with respect to $B\in [B_m(\tilde{c}), B_{m+1}(\tilde{c})[ \cap \prats$ 
if and only if
the SS $(T,\alpha(\tilde{c}),\tilde{1})$ is optimal for the 
C-model $M' = (T,\{\alpha(\tilde{c})\},{{I}})$ 
with respect to 
$B'\in [B_m(\alpha(\tilde{c})), B_{m+1}(\alpha(\tilde{c}))[ \cap \prats$.
\end{proposition}
Combining Propositions~\ref{prp:P-optimal} and~\ref{prp:C-optimal}, we
have the following summarizing corollary.
\begin{corollary}
\label{cor:alpha(CandP)}
Let $\alpha\in G_n$.

The SS $(T,\tilde{1},\tilde{p})$ is optimal for the P-model
$M = (T,{{I}},\{\tilde{p}\})$ if and only if
the SS $(T,\tilde{1}, \alpha(\tilde{p}))$ is optimal for the P-model
$M' = (T,{{I}},\{\alpha(\tilde{p})\})$.

The SS $(T,\tilde{c},\tilde{1})$ is optimal for the C-model
$M = (T,\{\tilde{c}\},{{I}})$ if and only if
the SS $(T, \alpha(\tilde{c}),\tilde{1})$ is optimal for the C-model
$M' = (T,\{\alpha(\tilde{p})\},{{I}})$.
\end{corollary}
Corollary~\ref{cor:alpha(CandP)} shows that optimality
of security systems of both C- and P-models is $G_n$-invariant
when applied to the prize and cost vector, respectively.

Recall the equivalence class $[\tilde{x}] = G_n(\tilde{x})\cap {\prats^n}$
from Definition~\ref{def:equiv-class}. We can now define induced 
equivalence classes of SS of both C- and P-models.
By Corollary~\ref{cor:alpha(CandP)} the following definition
is valid (that is, the terms are all well defined).
\begin{definition}
\label{def:CandP-class}
For a C-model $M = (T,C,{{I}})$ and a SS $(T,\tilde{c},\tilde{1})$
of $M$, we let
\[
[(T,\tilde{c},\tilde{1})] := 
  \{(T,\tilde{x},\tilde{1}) : \tilde{x}\in [\tilde{c}]\}.
\]
We say that $[(T,\tilde{c},\tilde{1})]$ is {\em optimal} if one
$(T,\tilde{x},\tilde{1})\in [(T,\tilde{c},\tilde{1})]$ is optimal
for its corresponding $M = (T,\{\tilde{x}\},{{I}})$, since then
each element in $[(T,\tilde{c},\tilde{1})]$ is also optimal.

Likewise, for a P-model $M = (T,{{I}},P)$ and a SS 
$(T,\tilde{1},\tilde{p})$ of $M$, we let 
\[
[(T,\tilde{1},\tilde{p})] := 
  \{(T,\tilde{1},\tilde{y}) : \tilde{y}\in [\tilde{p}]\}.
\]
We say that $[(T,\tilde{1},\tilde{p})]$ is {\em optimal} if one 
$(T,\tilde{1},\tilde{y})\in [(T,\tilde{1},\tilde{p})]$ is optimal
for its corresponding $M = (T,{{I}},\{\tilde{y}\})$, since
then each element in $[(T,\tilde{1},\tilde{p})]$ is also optimal.
\end{definition}
With the setup just presented we now can define the dual of both
vector classes and SS classes for C- and P-models in the following.
\begin{definition}
\label{def:dual}
For a vector $\tilde{x}$ and $[\tilde{x}] = G_n(\tilde{x})\cap {\prats^n}$
as in Definition~\ref{def:equiv-class}, then $[\tilde{x}]^* := [-\tilde{x}]$ 
is the {\em dual vector class} of $[\tilde{x}]$.

For a C-model $M = (T,C,{{I}})$ and a SS $(T,\tilde{c},\tilde{1})$
of $M$, then $[(T,\tilde{c},\tilde{1})]^* := [(T,\tilde{1},-\tilde{c})]$
is the corresponding 
{\em dual P-model security system class (dual P-model SS class)}
of the C-model class $[(T,\tilde{c},\tilde{1})]$.

Likewise, for a P-model $M = (T,{{I}},P)$ and a SS 
$(T,\tilde{1},\tilde{p})$ of $M$, then 
$[(T,\tilde{1},\tilde{p})]^* := [(T,-\tilde{p},\tilde{1})]$
is the corresponding 
{\em dual C-model security system class (dual C-model SS class)}
of the P-model class $[(T,\tilde{1},\tilde{p})]$.
\end{definition}
Note that the double-dual yields the original class in each case:
$[\tilde{x}]^{**} = [-\tilde{x}]^* = [\tilde{x}]$, and 
\[
[(T,\tilde{c},\tilde{1})]^{**} = [(T,\tilde{1},-\tilde{c})]^* 
  = [(T,\tilde{c},\tilde{1})], \ \ 
[(T,\tilde{1},\tilde{p})]^{**} = [(T,-\tilde{p},\tilde{1})]^* 
  = [(T,\tilde{1},\tilde{p})]. 
\]

For a P-model $M=(T,{{I}},P)$ and a SS P-model class 
$[(T,\tilde{1},\tilde{p})]$
we can always assume the prize vector $\tilde{p}$ is such  
$p_i\in [0,1]\cap\prats$ for each $i$, since 
$\alpha(\tilde{x}) = a\tilde{x}$ is indeed an element of $G_n$ for
any $a>0$. In this way 
$\tilde{c} = \tilde{1}-\tilde{p}\in ([0,1]\cap\prats)^n$
is a legitimate cost vector, and we have $[\tilde{p}]^* = [\tilde{1}-\tilde{p}]$
and $[(T,\tilde{1},\tilde{p})]^* = [(T,\tilde{1}-\tilde{p},\tilde{1})]$. 
In what follows, we will call such a prize vector {\em scaled}.
The following is easy to show.
\begin{claim}
\label{clm:sum-dual}
For a scaled prize vector $\tilde{p}$ with $p_i\in [0,1]\cap\prats$ 
for each $i$, and a rooted subtree $\tau$ of $T$ with $|E(\tau)| = m$, then 
$\pr(\tau,\tilde{1},\tilde{p}) + \cost(\tau,\tilde{1}-\tilde{p},\tilde{1}) = m$.
\end{claim}
Let $\tilde{p}$ be a scaled prize vector and assume $B$ is a budget
with $\maxp(B,\tilde{1}-\tilde{p},\tilde{1}) = m$. Then there is
a rooted subtree $\tau$ of $T$ on $m$ edges such that
$\cost(\tau,\tilde{1}-\tilde{p},\tilde{1}) \leq B$, and hence
there is such a $\tau$ of smallest cost. Hence, we may assume
$\tau$ is indeed such a rooted subtree of smallest cost.
By Claim~\ref{clm:sum-dual} applied to $\tilde{1}-\tilde{p}$, which is
also scaled, we then have
$\pr(\tau,\tilde{1},\tilde{p}) = m - \cost(\tau,\tilde{1}-\tilde{p},\tilde{1})$
with the smallest $\cost(\tau,\tilde{1}-\tilde{p},\tilde{1})$ among
rooted subtrees $\tau$ on $m$ edges, and hence  
$\pr(\tau,\tilde{1},\tilde{p})$ is maximum among all rooted subtrees
$\tau$ on $m$ edges, and so 
$\pr(\tau,\tilde{1},\tilde{p}) = \maxp(m,\tilde{1},\tilde{p})$. Hence,
\[
B\geq \cost(\tau,\tilde{1}-\tilde{p},\tilde{1}) = 
m - \pr(\tau,\tilde{1},\tilde{p}) = m - \maxp(m,\tilde{1},\tilde{p}).
\]
Since $\cost(\tau,\tilde{1}-\tilde{p},\tilde{1})$ is the smallest cost
among all rooted subtrees on $m$ edges, then 
\[
B'=\cost(\tau,\tilde{1}-\tilde{p},\tilde{1})=m-\maxp(m,\tilde{1},\tilde{p})
\]
is indeed the smallest cost with $\maxp(B',\tilde{1}-\tilde{p},\tilde{1}) = m$.
By Definition~\ref{def:B_m(c)} we then have the following.
\begin{lemma}
\label{lmm:B_m(1-p)}
For $m\in\{0,1,\ldots,n\}$ and a scaled (prize) vector $\tilde{p}$, we have
\[
B_m(\tilde{1}-\tilde{p}) = m - \maxp(m,\tilde{1},\tilde{p}).
\]
\end{lemma}
As a direct consequence of Lemma~\ref{lmm:B_m(1-p)}, we then have 
\begin{corollary}
\label{cor:B-vs-maxp}
For any $m\in\{0,1,\ldots,n\}$ and scaled vectors $\tilde{p}$ 
and $\tilde{p}'$, we have
\[
B_m(\tilde{1}-\tilde{p})\geq B_m(\tilde{1}-\tilde{p}')
\Leftrightarrow
\maxp(m,\tilde{1},\tilde{p}) \leq \maxp(m,\tilde{1},\tilde{p}').
\]
\end{corollary}
We can now prove one of the main results in this section.
\begin{theorem}
\label{thm:dual-opt}
Let $M = (T,{{I}},P)$ be a P-model, $(T,\tilde{1},\tilde{p})$ a SS
for $M$ where $\tilde{p}$ is scaled, and $m\in \{0,1,\ldots, n\}$.
Then $\maxp(m,\tilde{1},\tilde{p}) \leq \maxp(m,\tilde{1},\tilde{p}')$
for any $\tilde{p}'$ with $\{\tilde{p}'\} = P$ if and only if
$\maxp(B,\tilde{1}-\tilde{p},\tilde{1}) \leq 
\maxp(B,\tilde{1}-\tilde{p}',\tilde{1})$ for any budget $B$
with $\maxp(B,\tilde{1}-\tilde{p},\tilde{1}) = m$ and for any
$\tilde{p}'$ with $\{\tilde{p}'\} = P$. 
\end{theorem}
\begin{proof}
By Corollary~\ref{cor:B-vs-maxp} we have that
$\maxp(m,\tilde{1},\tilde{p}) \leq \maxp(m,\tilde{1},\tilde{p}')$
for any $\tilde{p}'$ with $\{\tilde{p}'\} = P$ if and only if
$B_m(\tilde{1}-\tilde{p})\geq B_m(\tilde{1}-\tilde{p}')$ for any
$\tilde{p}'$ with $\{\tilde{p}'\} = P$ which, 
by Proposition~\ref{prp:B_m(c)-equiv}, holds if and only if 
$\maxp(B,\tilde{1}-\tilde{p},\tilde{1})
\leq\maxp(B,\tilde{1}-\tilde{p}',\tilde{1})$ for all budgets $B$
with $\maxp(B,\tilde{1}-\tilde{p},\tilde{1}) = m$ and for
all $\tilde{p}'$ with $\{\tilde{p}'\} = P$. 
\end{proof}
Note that by Theorem~\ref{thm:dual-opt} we have that
$\maxp(B,\tilde{1},\tilde{p}) \leq \maxp(B,\tilde{1},\tilde{p}')$
for any budget $B$ and any $\tilde{p}'$ with $\{\tilde{p}'\} = \{\tilde{p}\}$,
if and only if
$\maxp(B,\tilde{1}-\tilde{p},\tilde{1}) 
\leq\maxp(B,\tilde{1}-\tilde{p}',\tilde{1})$ for any budget $B$ and any 
$\tilde{p}'$ with $\{\tilde{p}'\} = \{\tilde{p}\}$.
Hence, by Corollary~\ref{cor:alpha(CandP)} and 
Theorem~\ref{thm:dual-opt} we therefore have the main conclusion of this 
section in light of Definition~\ref{def:CandP-class}.
\begin{corollary}
\label{cor:dual-main}
For a rooted tree $T$ and a prize vector $\tilde{p}\in{\prats^n}$,
then $[(T,\tilde{1},\tilde{p})]$ is an optimal P-model SS class
if and only if the dual C-model SS class 
$[(T,\tilde{1},\tilde{p})]^* = [(T,-\tilde{p},\tilde{1})]$ is optimal.

In particular, if $\tilde{p}$ is scaled, then
the SS $(T,\tilde{1},\tilde{p})$ is optimal for 
the P-model $M = (T,{{I}},\{\tilde{p}\})$ if and only if 
the SS $(T,\tilde{1}-\tilde{p},\tilde{1})$ is optimal for 
the C-model $M = (T,\{\tilde{1}-\tilde{p}\},{{I}})$.
\end{corollary}
Consequently, by 
Corollary~\ref{cor:path-optimal}, 
Theorems~\ref{thm:star},~\ref{thm:SS-trees},~\ref{thm:3-cat-SS} 
and~\ref{thm:4-spid-SS}
and Corollary~\ref{cor:dual-main}, we have the following
summarizing result.
\begin{theorem}
\label{thm:C-P-main}
For a rooted tree $T$ on $n$ non-root vertices the following are equivalent:
\begin{enumerate}
  \item Any P-model $M=(T,{{I}},P)$ has an optimal SS.
  \item Any C-model $M=(T,C,{{I}})$ has an optimal SS.
  \item $T$ is one of the following types:
(i) a rooted path, 
(ii) a rooted star, 
(iii) a rooted 3-caterpillar,
or (iv) a rooted 4-spider.
\end{enumerate}
\end{theorem}
Note that by (\ref{eqn:order-preserving}) we have, in particular, that
each $\alpha\in G_n$ preserves the order of the entries of each
$\tilde{x}\in{\prats^n}$, so each $\tilde{x}\in [\tilde{p}]$ 
has the same order of its entries as $\tilde{p}$ does. But clearly,
the dual operation on $[\tilde{x}]^* = [-\tilde{x}]$ 
is order reversing, that is, we have 
that $x_i\leq x_j$ for any $\tilde{x}\in [\tilde{p}]$ 
if and only if 
$y_i\geq y_j$ for any $\tilde{y}\in [-\tilde{p}] = [\tilde{p}]^*$.
Since the optimal assignments of prizes from a given multiset $P$
are given in Theorems~\ref{thm:3-cat-SS} and~\ref{thm:4-spid-SS},
we then have by Corollary~\ref{cor:dual-main} the following
theorems for C-models as well.
\begin{theorem}
\label{thm:3-cat-CSS}
Let $M = (T,C,{{I}})$ be a C-model where $T$ is a rooted 3-caterpillar 
as in (\ref{eqn:3-cat-label}) and
$C = \{c_1,\ldots,c_n\}$ is a multiset of possible edge-costs indexed
decreasingly $c_1\geq c_2\geq\cdots\geq c_n$. Then the SS $(T,c,\mathbf{1})$,
where $c(e_i) = c_i$ for each $i\in \{1,\ldots,n\}$ is an optimal SS for $M$.
\end{theorem}
\begin{theorem}
\label{thm:4-spid-CSS}
Let $M = (T,C,{{I}})$ be a P-model, where $T$ is a rooted 4-spider 
as in (\ref{eqn:4-spider-label}) and 
$C = \{c_1,\ldots,c_n\}$ is a multiset of possible edge-costs indexed
decreasingly $c_1\geq c_2\geq\cdots\geq c_n$. Then the SS $(T,c,\mathbf{1})$,
where $c(e_i) = c_i$ for $i\in \{1,\ldots,k\}$ and 
$c(e_i) = c_{n+k+1-i}$ for $i\in\{k+1,\ldots,n\}$ is an optimal SS for $M$.
\end{theorem}

\section{Summary and Conclusions}
\label{sec:summary}

This paper defined a cyber-security model to explore defensive
security systems.  The results obtained mathematically support the
intuition that it is best to place stronger defenses in the outer
layers and more-valuable prizes in the deeper layers.  We defined
three types of SSs: improved, good, and optimal.  We showed that it is
not always possible to find an optimal SS for a given CSM, but
demonstrated for rooted paths and stars that optimal SSs do exist.  The
results mathematically show that a path produces the best
cyber-security, however, burying something $n$ levels deep for large
$n$ may prevent the friendly side from accessing the ``information''
effectively.  The results show, in general, that trees having greater
depth provide more security in this setting. 

We showed the any CSM
is equivalent to a CSM where either all the edge penetration costs
are unit priced (a P-model) or where all the vertices have a unit
prize (C-model), by allowing larger underlying rooted trees. 
We then characterised for which trees a P-model
has an optimal SSs, and we also did that for the C-models. We noted that
the P- and C-models have optimal SSs for exactly the same types
of rooted trees. This was then explained by obtaining a duality
between the P- and C-models in the penultimate section of the paper.

We gave an $O(n\log n)$ algorithm for producing a good SS that was
based on sorting. It is not clear how strong such a good SS is, as
there may be many such good SSs, and some may be better than others.
It would be interesting to come up with a comparison metric to rank
various good SSs.  We must continue to explore models of
cyber-security systems to develop the foundations needed to combat the
ongoing and increasing number of cyber attacks.  This work is but one
step in that direction.

We conclude the paper with a number of questions.
\begin{enumerate}
  \item Can we find an efficient algorithm to develop optimal SSs in the
cases where all penetration costs or all targets are from a finite
set of possible values? Say, if we have two possible penetrations costs
or three? Similarly for prizes?
  \item  In a two-player version of the model, what would be the best strategy
for a defender who is allowed to reposition a prize or a portion of a
prize after each move by an attacker?  And, what would the complexity
of this problem be?
  \item Are there on-line variants of the model that are interesting to study?
For example, a version where the topology of the tree changes
dynamically or where only a partial description is known to the
attacker.
  \item Could a dynamic programming approach be used to obtain a SS that were
somehow quantifiably better than a good SS or allow us to pick the
``best'' good SS?
  \item Is there a more-useful definition of neighboring configuration that
could lead to an efficient algorithm for producing better SSs, for
example, perhaps a definition where sibling vertices or edges can have
their prizes or penetration costs swapped, respectively?
\end{enumerate}

\section*{Acknowledgments}

This work was supported by the Office of Naval Research.  The work was
also supported by Thailand Research Fund grant No.~RSA5480006.

\end{document}